\documentclass[11pt,twoside]{article}
\usepackage{fancyhdr}
\usepackage[colorlinks,citecolor=blue,urlcolor=blue,linkcolor=blue,bookmarks=false]{hyperref}
\usepackage{amsfonts,epsfig,graphicx}
\usepackage{afterpage}
\usepackage{comment}
\usepackage{amsmath,amssymb,amsthm} 
\usepackage{fullpage}
\usepackage[T1]{fontenc} 
\usepackage{epsf} 
\usepackage{graphics} 
\usepackage{amsfonts,amsmath}
\usepackage{mathtools}
\usepackage{natbib} 
\usepackage{psfrag,xspace}
\usepackage{color,etoolbox}
\usepackage{xcolor}
\usepackage{subcaption} 
\usepackage{dsfont}
\usepackage[page]{appendix}
\usepackage{soul} 

\def\indist{\rightsquigarrow}
\def\ind{\perp\!\!\!\perp}

\newcommand{\Pb}{\mathbb{P}}
\newcommand{\Pn}{\mathbb{P}_n}

\newcommand{\E}{\mathbb{E}}
\newcommand{\R}{\mathbb{R}}

\newcommand{\bX}{\mathbf{X}}
\newcommand{\bx}{\mathbf{x}}

\DeclareMathOperator*{\argmin}{arg\,min}
\DeclareMathOperator*{\argmax}{arg\,max}

\DeclareSymbolFont{bbold}{U}{bbold}{m}{n}
\DeclareSymbolFontAlphabet{\mathbbold}{bbold}

\newtheorem{theorem}{Theorem}
\newtheorem{lemma}{Lemma}
\newtheorem{corollary}{Corollary}
\newtheorem{proposition}{Proposition}

\newtheorem{assumption}{Assumption}
\newtheorem{remark}{Remark}
\theoremstyle{definition}

\theoremstyle{remark}

\begin{document}

\def\spacingset#1{\renewcommand{\baselinestretch}%
{#1}\small\normalsize} \spacingset{1}

\raggedbottom
\allowdisplaybreaks[1]


  \title{\vspace*{-.4in} {Covariate-assisted bounds on causal effects with instrumental variables}}
  \author{Alexander W. Levis$^1$, Matteo Bonvini$^1$$^*$, Zhenghao Zeng$^1$\thanks{MB and ZZ contributed equally.},
  \\ Luke Keele$^2$, Edward H. Kennedy$^1$ \\  \\ \\
    $^1$Department of Statistics \& Data Science, \\
    Carnegie Mellon University \\
    $^2$Department of Surgery, \\
    University of Pennsylvania \\ \\ 
    \texttt{\{alevis, mbonvini, zhenghaz\} @ andrew.cmu.edu}; \\ \texttt{luke.keele@gmail.com}; \texttt{edward@stat.cmu.edu} \\
\date{}
    }
    
  \maketitle
  \thispagestyle{empty}

\begin{abstract}
When an exposure of interest is confounded by unmeasured factors, an instrumental variable (IV) can be used to identify and estimate certain causal contrasts. Identification of the marginal average treatment effect (ATE) from IVs relies on strong untestable structural assumptions. When one is unwilling to assert such structure, IVs can nonetheless be used to construct bounds on the ATE. Famously, \cite{balke1997bounds} proved tight bounds on the ATE for a binary outcome, in a randomized trial with noncompliance and no covariate information. We demonstrate how these bounds remain useful in observational settings with baseline confounders of the IV, as well as randomized trials with measured baseline covariates. The resulting bounds on the ATE are non-smooth functionals, and thus standard nonparametric efficiency theory is not immediately applicable. To remedy this, we propose (1) under a novel margin condition, influence function-based estimators of the bounds that can attain parametric convergence rates when the nuisance functions are modeled flexibly, and (2) estimators of smooth approximations of these bounds. We propose extensions to continuous outcomes, explore finite sample properties in simulations, and illustrate the proposed estimators in an observational study targeting the effect of higher education on wages.
\end{abstract}

\bigskip


\pagebreak

\section{Introduction}

A primary goal in many scientific endeavors is to determine whether an intervention has a causal effect on an outcome. However, simple comparisons between the outcomes of treated and control groups are often complicated by confounding: differences in outcomes between those who are and are not treated due to pre-treatment differences rather than the effect of the treatment itself. One solution to selection bias of this form is to allocate the treatment via randomization, rendering pre-treatment distributions equal, on average. In many circumstances, however, randomization is infeasible or unethical. When this is the case, one alternative is to identify an instrumental variable (IV). 
For a variable to be an IV, it must meet the three following conditions: (a) the IV must be associated with the exposure; (b) the IV must be randomly or as-if randomly assigned; and (c) the IV cannot have a direct effect on the outcome \citep{Angrist:1996}. Under these conditions, as well as some further structural assumptions, an IV can provide a consistent estimate of a causal effect even in the presence of unobserved confounding between the exposure and the outcome. See \citet{baiocchi2014} and \citet{imbens2014instrumental} for general reviews of IV methods. 

In most analyses, investigators focus on point identification--asserting sufficient structure so that a single parameter describing a causal effect can be expressed as a function of the observed data distribution. In fact, assumptions (a)--(c), on their own, are insufficient to point identify an average causal effect. Point identification requires investigators to assume either some form of homogeneity (e.g., lack of effect modification by unmeasured confounders) or an assumption known as monotonicity \citep{robins1994correcting, hernan2006instruments,tan2010marginal,angrist1996identification}. Critically, these structural IV assumptions are untestable and may be controversial in many applications. One alternative approach is to relax key IV assumptions using partial identification. Under the partial identification approach, analysts seek to estimate bounds on the parameter of interest, which can typically be done with weaker assumptions \citep{Manski:1990,Manski:1995}. There is a large body of work on partial identification in IV designs with foundational work done by, e.g., \citet{balke1997bounds} and \citet{manski2009more,manski2000monotone}. See \cite{swanson2018partial} for a general overview of partial identification approaches to an IV analysis. 

There exists a substantial literature studying bounds on causal effects, incorporating covariate adjustment under various sets of assumptions. Broadly, use of covariates is beneficial in either (a) experimental settings, where baseline covariate adjustment can improve efficiency \citep{yang2001efficiency, tsiatis2008covariate} and bound width \citep{cai2007non}, or (b) observational settings, where confounder adjustment is necessary for key assumptions to hold \citep{grilli2008nonparametric, feller2017principal}. Notably, \cite{cai2007non} illustrated possible improvements on natural bounds \citep{manski1990} and tight bounds \citep{balke1997bounds} on average treatment effects in a randomized trial using covariate adjustment. \cite{long2013sharpening} showed covariate-assisted bounds on principal effects are narrower than the unadjusted bounds. Such adjusted bounds are widely used in principal stratification analysis on evaluating the effects of training and school programs \citep{lee2005training, miratrix2018bounding}. However, most existing work considers the discrete covariate setting (hence plug-in style empirical average estimators are applicable) and does not explore efficiency theory on the bounds. \cite{grilli2008nonparametric} examined adjusted principal stratification-based bounds using continuous covariates under strong parametric assumptions that can be hard to justify in practice.

In this paper, we extend the Balke-Pearl bounds to include baseline covariate information to remedy these issues. Our proposed estimators are based on semiparametric efficiency theory and use influence functions, which allow for flexible and efficient estimation via a variety of nonparametric machine learning based methods. The use of influence functions also allows us to derive simple closed-form variance estimators that are consistent. Importantly, our work may provide useful insights into estimation strategies for functionals that involve non-smooth functions on observed data and are not pathwise-differentiable without further assumptions. Such functionals commonly appear when targetting bounds on causal effects \citep{sachs2022general, grilli2008nonparametric, long2013sharpening}.

The remainder of the paper is organized as follows. In Section~\ref{sec:background}, we define notation, review key IV assumptions, and outline our target causal estimand: the average treatment effect. In Section~\ref{sec:bounds}, we review extant methods for bounding the average treatment effect (ATE) based on an IV. We also provide a simple illustration to demonstrate the benefits of incorporating covariate information, and present a general result to guide which covariates to include in the ensuing analyses. Importantly, the nonparametric bounds outlined in Section~\ref{sec:bounds} are non-smooth functionals of the observed data distribution, so standard semiparametric efficiency theory does not immediately apply. In view of this challenging setting, Sections~\ref{sec:binary}, \ref{sec:smooth}, and \ref{sec:continuous} detail the three main methodological contributions of the paper: (i) efficient estimators of the true non-smooth bound functionals under a margin condition (in Section~\ref{sec:binary}); (ii) valid but conservative bounds on the ATE based on estimators of smooth functional approximations to the bounds (in Section~\ref{sec:smooth}); and (iii) extensions of the proposed methods to the case of a general bounded outcome, discrete or continuous (in Section~\ref{sec:continuous}). In Section~\ref{sec:simulation}, we present a simulation study to investigate the finite sample properties of the proposed estimators. Finally, in Section~\ref{sec:data}, we apply our proposed methods to data from an observational study to assess the effects of college education on wages. Proofs of all results can be found in the appendices.

\section{Background} 
\label{sec:background}

\subsection{Notation and Assumptions}
Consider the standard instrumental variable setup, in which the observed data are $n$ iid copies of
$O = (\boldsymbol{X},Z, A, Y) \sim \Pb$, where $\boldsymbol{X} \in \mathcal{X} \subseteq \mathbb{R}^d$ is a vector of covariates, and $Z, A \in \{0,1\}$ are binary instrument and exposure variables, respectively. Our focus is mostly on binary outcomes, $Y \in \{0,1\}$; we will discuss extensions for non-binary outcomes in Section~\ref{sec:continuous}. We let $A(z)$ and $Y(z)$ denote the counterfactual exposure and outcome values, had the instrument been set to $Z = z$, for $z \in \{0,1\}$. Similarly, we also define $Y(a)$ and $Y(z,a)$ to be the potential outcomes under an intervention that sets $A = a$, and an intervention that sets both $Z = z$ and $A = a$, respectively. For $z \in \{0,1\}$, let $\lambda_z(\boldsymbol{X})=\Pb[Z=z \mid \boldsymbol{X}]$. Next, we review the set of assumptions that we take as given in the analysis. First, we make the two following assumptions:

\begin{assumption}[Consistency] \label{ass:consistency}
    $A(Z) = A$, and $Y(Z, A) = Y(A) = Y$, almost surely.
\end{assumption}

\begin{assumption}[Positivity] \label{ass:positivity}
    For some $\epsilon > 0$, $\lambda_1(\boldsymbol{X}) \in [\epsilon, 1 - \epsilon]$, almost surely.
\end{assumption}

In words, Assumption~\ref{ass:consistency} asserts that interventions on $Z$ and $A$ are well-defined, and that there is no interference between subjects, so that a unit's potential treatment and outcome can be linked to their observed variables; Assumption~\ref{ass:positivity} asserts that either instrument value can be realized, over all strata determined by $\boldsymbol{X}$. 
These two assumptions are not unique to the IV framework, and are commonly invoked in many causal inference settings.  Next, we make the following ``core'' IV assumptions:


\begin{assumption}[Unconfoundedness] \label{ass:UC}
    $Z \ind (A(z), Y(z)) \mid \boldsymbol{X}$.
\end{assumption}

\begin{assumption}[Exclusion Restriction] \label{ass:ER}
    $Y(z, a) \equiv Y(a)$, for $z,a \in \{0,1\}$.
\end{assumption}

Assumption~\ref{ass:UC} asserts that the effect of $Z$ on $A$ and $Y$ is unconfounded, given measured covariates $\boldsymbol{X}$. Note that in certain special cases, Assumption~\ref{ass:UC} holds unconditionally, and we can identify the IV effect without conditioning on $\boldsymbol{X}$.  For example, when $Z$ is marginally randomized (e.g., in a trial), Assumption~\ref{ass:UC} holds by design. Alternatively, in certain natural experiments, analysts may assert that Assumption~\ref{ass:UC} holds unconditionally, e.g. \citet{angrist1996children}. 
Assumption~\ref{ass:ER} asserts that the effect of $Z$ on $Y$ acts entirely through its effect on $A$, i.e., $Z$ has no direct effect on $Y$. 
Note that most IV studies adopt an assumption of non-zero association between $Z$ and $A$ (often referred to as a ``relevance'' assumption). However, this assumption is not formally required for partial identification of the average treatment effect.

\subsection{Estimands and Additional Structural Assumptions}

In what follows, we focus on the average treatment effect (ATE) as the target causal estimand:
\[
\mathbb{E}\left(Y(a = 1) - Y(a = 0)\right).
\]
Critically, under the four assumptions introduced in the previous section, the ATE is not point identified. Analysts typically take one of two approaches for point identification. The first approach invokes some type of homogeneity assumptions and places various restrictions on how the effects of $A$ and $Z$ vary from unit to unit in the study population. See \citet{hernan2019} and \citet{wang2018} for prominent examples. 
However, homogeneity assumptions are often implausible or difficult to verify in specific applications. The second approach invokes an assumption known as monotonicity, which has the following form: $A(z = 1) \geq A(z = 0)$, i.e., if $A(z = 0) = 1$ then $A(z = 1) = 1$ \citep{imbens1994}. Under monotonicity, the target estimand is no longer the ATE, but instead is the local average treatment effect (LATE):
\begin{equation} \label{eq:LATE}
\mathbb{E}(Y(a = 1) - Y(a = 0) \mid A(z = 1) > A(z = 0)) 
\end{equation}
\noindent 
Here, investigators must be content with a more local estimand that may not generalize to the ATE. In this paper, we only assume that Assumptions~\ref{ass:consistency}--\ref{ass:ER} hold, and aim to construct tight bounds on the ATE parameter.

\section{Partial Identification for IVs and the Role of Covariate Information}
\label{sec:bounds}

Next, we provide an overview of one partial identification approach for the IV framework. We then provide a brief illustration to demonstrate how covariate information can alter the bounds on the ATE. We also present theoretical results to elucidate when we may expect the inclusion of covariates to tighten the bounds on the ATE.

\subsection{Review: Balke-Pearl Bounds} 
\label{sec:BP-bounds}

In their seminal work, Alexander Balke and Judea Pearl leveraged symbolic linear programming to develop sharp nonparametric bounds on the ATE for a binary outcome \citep{balke1994, balke1995, balke1997bounds}. Notably, their bounds only invoke Assumptions~\ref{ass:consistency}--\ref{ass:ER}, and are provably tight under these assumptions. In both applications and simulations they demonstrated substantial narrowing of the bounds on the ATE compared to partial identification results in \citet{robins1989} and \citet{manski1990}.

The results in \cite{balke1997bounds} focused exclusively on marginally randomized instruments such that $\boldsymbol{X} = \emptyset$. As such, their methods do not make use of any measured covariates. Nonetheless, their results hold just as well for bounding the conditional average treatment effect (CATE), $\mathbb{E}[Y(a = 1) - Y(a = 0) \mid \boldsymbol{X} = \boldsymbol{x}]$, for any $\boldsymbol{x} \in \mathcal{X}$, in two important settings where Assumptions~\ref{ass:consistency}--\ref{ass:ER} may hold: (i) an experiment with $Z$ marginally randomized and baseline covariates $\boldsymbol{X}$ are measured, and (ii) an observational setting where $\boldsymbol{X}$ are baseline confounders required for $Z$ to be a valid IV. In case (i), we will demonstrate that incorporating baseline covariate information $\boldsymbol{X}$ can both provide tighter theoretical bounds, and improve statistical precision. In case (ii), conditioning on $\boldsymbol{X}$ is required for the IV assumptions to hold, and thus for the bounds to be valid.

To begin, we review the main theoretical result from \cite{balke1997bounds}. For each $y, a, z \in \{0,1\}$, define $\pi_{ya.z}(\boldsymbol{X}) = \Pb(Y=y, A = a \mid \boldsymbol{X}, Z= z)$. Moreover, define the two 8-dimensional vectors $\boldsymbol{\theta}_{\ell}(\boldsymbol{X}) = (\theta_{\ell, 1}(\boldsymbol{X}), \ldots, \theta_{\ell, 8}(\boldsymbol{X})), \boldsymbol{\theta}_{u}(\boldsymbol{X}) = (\theta_{u, 1}(\boldsymbol{X}), \ldots, \theta_{u, 8}(\boldsymbol{X})) \in [-1,1]^8$, where omitting inputs,
\begin{equation}\label{eq:lower-bound}
\begin{aligned}
    \theta_{\ell, 1} &= \pi_{11.1} +
  \pi_{00.0} - 1 \\
  \theta_{\ell, 2} &= \pi_{11.0} +
  \pi_{00.1} - 1 \\
  \theta_{\ell, 3} &= - \pi_{01.1} -
  \pi_{10.1} \\
  \theta_{\ell, 4} &= - \pi_{01.0} -
  \pi_{10.0} \\
  \theta_{\ell, 5} &= \pi_{11.0} -
  \pi_{11.1} - \pi_{10.1} -
  \pi_{01.0} - \pi_{10.0} \\
  \theta_{\ell, 6} &= \pi_{11.1} -
  \pi_{11.0} - \pi_{10.0} -
  \pi_{01.1} - \pi_{10.1} \\
  \theta_{\ell, 7} &= \pi_{00.1} -
  \pi_{01.1} - \pi_{10.1} -
  \pi_{01.0} - \pi_{00.0} \\
  \theta_{\ell, 8} &= \pi_{00.0} -
  \pi_{01.0} - \pi_{10.0} -
  \pi_{01.1} - \pi_{00.1}
\end{aligned}
\end{equation}

and
\begin{equation}\label{eq:upper-bound}
    \begin{aligned}
\theta_{u, 1} &= 1 - \pi_{01.1} -
  \pi_{10.0} \\
  \theta_{u, 2} &= 1 - \pi_{01.0} -
  \pi_{10.1} \\
  \theta_{u, 3} &= \pi_{11.1} +
  \pi_{00.1}\\
  \theta_{u, 4} &= \pi_{11.0} +
  \pi_{00.0}\\
  \theta_{u, 5} &= -\pi_{01.0} +
  \pi_{01.1} + \pi_{00.1} +
  \pi_{11.0} + \pi_{00.0} \\
  \theta_{u, 6} &= -\pi_{01.1} +
  \pi_{11.1} + \pi_{00.1} +
  \pi_{01.0} + \pi_{00.0} \\
  \theta_{u, 7} &= -\pi_{10.1} +
  \pi_{11.1} + \pi_{00.1} +
  \pi_{11.0} + \pi_{10.0}\\
  \theta_{u, 8} &= -\pi_{10.0} +
  \pi_{11.0} + \pi_{00.0} +
  \pi_{11.1} + \pi_{10.1} . 
\end{aligned}
\end{equation}

\noindent Finally, define $\gamma_{\ell}(\boldsymbol{X}) = \max_{1 \leq j \leq 8} \theta_{\ell, j}(\boldsymbol{X})$ and $\gamma_{u}(\boldsymbol{X}) = \min_{1 \leq j \leq 8} \theta_{u, j}(\boldsymbol{X})$. Balke and Pearl proved (e.g., see the main result in \citet{balke1997bounds}, p. 1173) that the CATE is bounded between $\gamma_{\ell}$ and $\gamma_u$, the tightest possible lower and upper bounds, respectively, under Assumptions~\ref{ass:consistency}--\ref{ass:ER}. We summarize their result in the following theorem.

\begin{theorem}[\citet{balke1997bounds}] \label{thm:BP}
    Under Assumptions~\ref{ass:consistency}--\ref{ass:ER}, for $(Z, A, Y) \in \{0,1\}^3$, the CATE can be bounded as
    \[\gamma_{\ell}(\boldsymbol{X}) \leq \mathbb{E}(Y(a = 1) - Y(a = 0) \mid \boldsymbol{X}) \leq \gamma_u(\boldsymbol{X}).\]
    These bounds are tight in the nonparametric model. Moreover, marginalizing yields bounds on the ATE: 
    \begin{equation}\label{eq:ATE-bounds}
\mathbb{E}_{\Pb}(\gamma_{\ell}(\boldsymbol{X})) \leq \mathbb{E}(Y(a = 1) - Y(a = 0)) \leq \mathbb{E}_{\Pb}(\gamma_u(\boldsymbol{X})),
\end{equation}
which are also tight in the nonparametric model.
\end{theorem}

\noindent
The primary goal of this paper is to construct efficient estimators of the tight bounds on the ATE given in Theorem~\ref{thm:BP}, $\mathcal{L}(\Pb) \coloneqq \mathbb{E}_{\Pb}(\gamma_{\ell}(\boldsymbol{X}))$ and $\mathcal{U}(\Pb) \coloneqq \mathbb{E}_{\Pb}(\gamma_{u}(\boldsymbol{X}))$. Statistically, this is a challenging task, since these are means of a pointwise maximum and minimum, each of which is a non-smooth function. We will outline two strategies to deal with this difficulty: (1) invoking a margin condition (see Section~\ref{sec:binary}), or (2) targeting a smooth approximation to the Balke-Pearl bounds (see Section~\ref{sec:smooth}). 

Before describing strategies for estimation, we first demonstrate the benefit of pursuing covariate-assisted bounds. In Section~\ref{sec:illustration}, we give a simple illustration in the randomized experimental setting, showing that covariate-assisted bounds can be substantially narrower than unadjusted bounds. In Section~\ref{sec:width}, we provide a general result to guide the choice of adjustment set $\boldsymbol{X}$ when more than one is possible.

\subsection{Motivating Illustration} 
\label{sec:illustration}

Consider a hypothetical randomized experiment with arm assignment $Z \sim \mathrm{Bernoulli}(0.5)$, independent of baseline covariates $X_1 \sim \mathrm{Bernoulli}(0.7)$, $X_2 \sim \mathrm{Unif}(-1,1)$, with $X_1 \ind X_2$. In this example, $X_1$ represents a behavioral or demographic factor, and $X_2$ represents an underlying risk score (i.e., low values represent good health and low risk, high values represent poor health and high risk). In this experiment, we suppose that there is a degree of noncompliance, which is completely determined by the covariates $X_1$ and $X_2$. Specifically, we set:
\begin{align*}
    A(z = 0)A(z = 1) &= \mathds{1}(X_2 \geq 0.99) \\
    \{1 - A(z = 0)\}\{1 - A(z = 1)\} &= \mathds{1}(X_2 \leq -0.99) \\
    A(z = 0)\{1 - A(z=1)\} &= (1 - X_1)\mathds{1}(X_2 \in (-0.5, 0.5]) \\
    \{1 - A(z = 0)\}A(z=1) &= 1 - \mathds{1}(|X_2| \geq 0.99) - (1 - X_1)\mathds{1}(X_2 \in (-0.5, 0.5])
\end{align*}
In words, we can define four principal strata: ``always takers'', who are exposed regardless of instrument status, are those with the very highest underlying risk; ``never takers'', who are not exposed regardless of instrument status, are those with the very lowest risk; ``defiers'', whose exposure value is opposite to that of the instrument, are those with intermediate risk and behavioral factor $X_1 = 0$; and ``compliers'', whose exposure value matches that of the instrument, comprise the remainder of the population \citep{Angrist:1996,Frangakis:2002}.
We further suppose the potential outcomes $Y(a)$ are completely determined by the compliance classes $U = (A(z = 0), A(z = 1))$, and set $Y(a) \mid U \sim \mathrm{Bernoulli}(p_a(U))$, with
\[p_a = \begin{cases}
    0.20 + 0.10a, & \text{ if } A(z = 0)A(z = 1) = 1, \\
    0.90 + 0.05a, & \text{ if } \{1 - A(z = 0)\}\{1 - A(z = 1)\} = 1, \\
        0.65 + 0.05a, & \text{ if } A(z = 0)\{1 - A(z=1)\} = 1, \\
    0.25 + 0.10a, & \text{ if } \{1 - A(z = 0)\}A(z=1) = 1.
\end{cases}.\]

Figure~\ref{fig:toy-example} shows the covariate-agnostic and covariate-adjusted Balke-Pearl bounds on the ATE resulting from Theorem~\ref{thm:BP}. The true bounds are plotted in blue and estimated 95\% confidence intervals (from one simulated sample of size 5,000) for the bounds are plotted in green. The covariate-agnostic bounds, while simple to compute as the maximum and minimum of 8 (true or empirical) probabilities, are quite wide and cover the null treatment effect of zero. The covariate-adjusted bounds, on the other hand, are very narrow and are bounded away from zero. Employing the estimator we propose in Section~\ref{sec:binary}, with flexible regression tree-based nuisance function estimation, we obtain valid and narrow estimated bounds on the ATE.

\begin{figure}[ht]
  \centering
  \includegraphics[width = \linewidth]{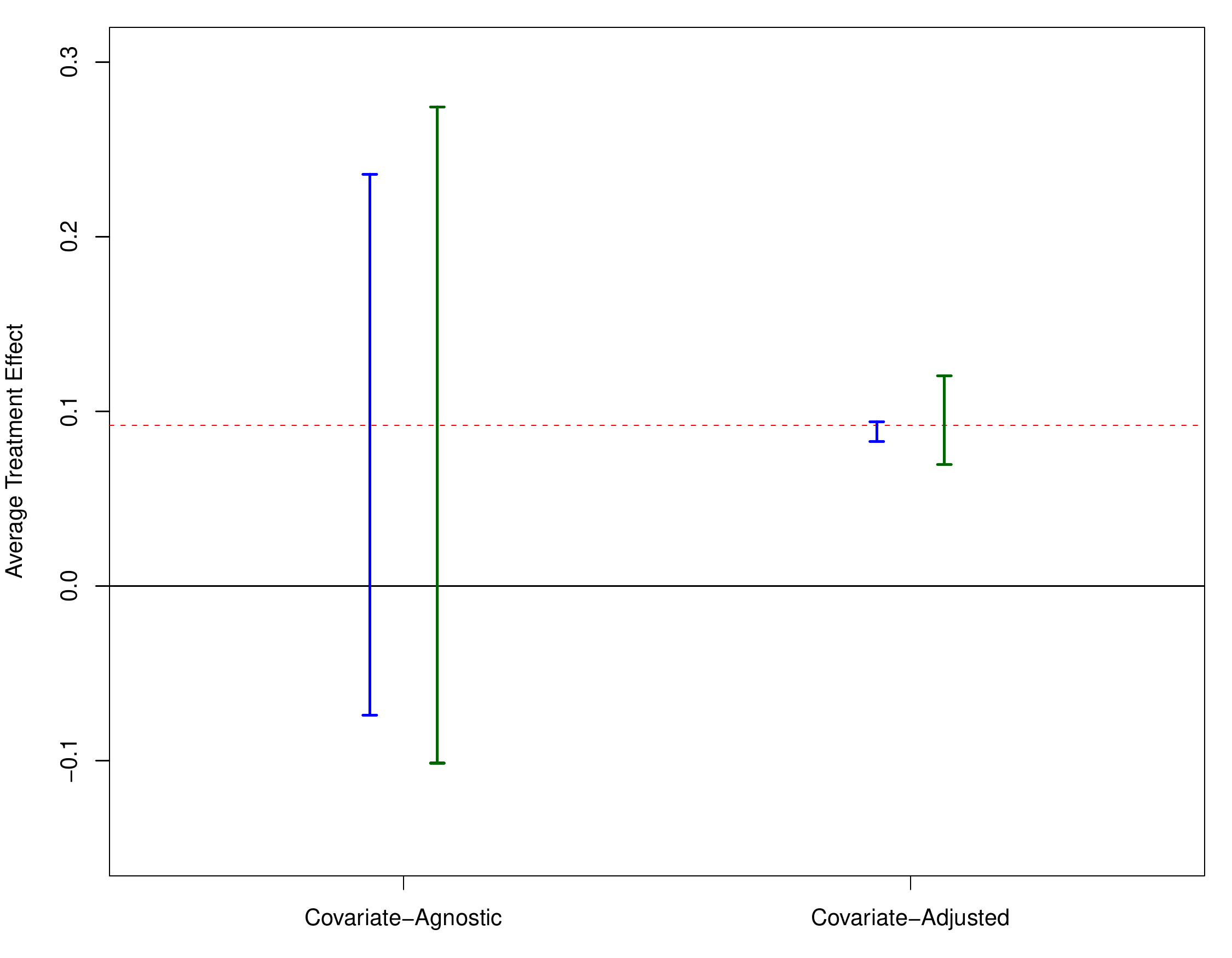}
  \caption{Balke-Pearl bounds for the ATE, with and without covariate adjustment. The red dashed line represents the true ATE. Blue bars represent the true theoretical bounds, and green bars are estimated 95\% confidence intervals from a single sample of 5,000 subjects (with 10-fold cross-fitting for adjusted approach). The solid black line indicates the reference value of zero average effect.} \label{fig:toy-example}
\end{figure}

Note that in this example, the covariates $X_1$, $X_2$ happen to comprise all confounders of the relationship between $A$ and $Y$, so that direct adjustment (e.g., with the $g$-formula) would identify the ATE. A practitioner without this knowledge, but with the foresight to measure $X_1$ and $X_2$, could still use these covariates to construct informative bounds. We also note that our example is somewhat extreme, as we would not typically expect the covariate-adjusted bounds to always be so narrow. However, we show below that in this experimental design, the covariate-adjusted bounds are guaranteed to be at least as narrow as unadjusted bounds, even in the presence of residual exposure-outcome confounding given $\boldsymbol{X}$. Indeed, one will want to adjust for $\boldsymbol{X}$ whenever these covariates are predictive of $A$ and $Y$.

\subsection{Width of the Bounds}
\label{sec:width}

In the illustration above, both choices of adjustment sets (no covariates or $(X_1, X_2)$) result in valid bounds since Assumptions~\ref{ass:consistency}--\ref{ass:ER} are satisfied with either choice. In our example, the covariate-assisted Balke-Pearl bounds improve substantially upon na\"ive covariate-agnostic bounds. In fact, such additional adjustment in general cannot do worse, and can only improve upon unadjusted bounds. We formalize this idea in the following result.

\begin{proposition} \label{prop:width}
    Suppose $\boldsymbol{X}$ renders $Z$ a valid instrument according
  to Assumptions~\ref{ass:consistency}--\ref{ass:ER}. Suppose
  that Assumptions~\ref{ass:consistency}--\ref{ass:ER} also hold after augmenting
  $\boldsymbol{X}$ with $\boldsymbol{G}$, and
  $Z \ind \boldsymbol{G} \mid \boldsymbol{X}$, i.e.,
  $\boldsymbol{G}$ doesn't predict $Z$, except possibly through
  association with $\boldsymbol{X}$. Then Balke-Pearl bounds based on
  $(\boldsymbol{X}, \boldsymbol{G})$ are at least as narrow as those
  based on $\boldsymbol{X}$, and may be strictly narrower.
\end{proposition}

According to Proposition~\ref{prop:width}, we would want to include in $\boldsymbol{X}$ any pure predictors of $(A, Y)$ that are not direct confounders of the effect of $Z$ on $A$ and $Y$, when constructing Balke-Pearl bounds. This is somewhat analogous to Lemma 4 in \citet{rotnitzky2020b}, which says that influence function-based estimation of the ATE from an observational study is improved by adding pure outcome predictors---so-called ``precision variables''---to an already sufficient set of confounders. Note, however, that in that setting improvement corresponded to lower variance, whereas here we are concerned both with the width of the theoretical bounds, as well as the variance of our estimators. Specializing Proposition~\ref{prop:width} to the case where $Z$ is marginally randomized yields the following corollary:

\begin{corollary}\label{cor:width}
    Suppose $Z \ind \boldsymbol{X}$ and Assumptions~\ref{ass:consistency}--\ref{ass:ER} hold unconditionally as well as given $\boldsymbol{X}$. Then Balke-Pearl bounds based on $\boldsymbol{X}$ are at least as narrow as unadjusted bounds.
\end{corollary}

Corollary~\ref{cor:width} explains what we observed in the motivating example of Section~\ref{sec:illustration}. That is, a marginally randomized IV guarantees that inclusion of baseline covariates $\boldsymbol{X}$ will result in bounds that are no worse than unadjusted bounds. The degree to which such inclusion narrows the bounds is of great interest, and will help in deciding which covariates to measure in practice---a more precise characterization of this improvement we leave for future research. For the remainder of the paper, we focus on constructing valid and efficient estimators of the covariate-assisted Balke-Pearl bounds.


\section{Bound Estimation Under a Margin Condition} 
\label{sec:binary}

As mentioned in Section~\ref{sec:BP-bounds}, the bound functionals $\mathcal{L}(\mathbb{P}) = \mathbb{E}_{\mathbb{P}}(\gamma_{\ell}(\boldsymbol{X}))$ and $\mathcal{U}(\mathbb{P}) = \mathbb{E}_{\mathbb{P}}(\gamma_{u}(\boldsymbol{X}))$ are not pathwise differentiable without further restrictions, and therefore do not have influence functions to enable estimation as described above. This is because the pointwise maximum, $\gamma_{\ell}(\boldsymbol{X}) = \max_{1 \leq j \leq 8} \theta_{\ell, j}(\boldsymbol{X})$, and pointwise minimum, $\gamma_{u}(\boldsymbol{X}) = \min_{1 \leq j \leq 8} \theta_{u, j}(\boldsymbol{X})$, are not differentiable everywhere. Indeed, we should not expect it to be possible in the nonparametric model to estimate the covariate-adjusted bounds at parametric rates. In this section, we introduce an additional assumption known as a margin condition that renders the bounds pathwise differentiable. Further, we describe estimators of the bounds that exploit the margin condition in order to achieve faster rates. For a broader overview of the use of influence functions for efficient and robust estimation, see Appendix~\ref{sec:appendix-smooth}.

\subsection{An infeasible estimator}\label{sec:infeasible}
To motivate the proposed estimator, we first consider an infeasible estimator of the bounds. Let $d_{\ell}(\boldsymbol{X}) \in \argmax_{1 \leq j \leq 8} \theta_{\ell, j}(\boldsymbol{X})$ and $d_{u}(\boldsymbol{X}) \in \argmin_{1 \leq j \leq 8} \theta_{u, j}(\boldsymbol{X})$. The bounds can thus be written as
\begin{align*}
& \mathcal{L}(\Pb) = \sum_{j = 1}^8 \E\left[ \{d_\ell(\boldsymbol{X}) = j\} \theta_{\ell, j}(\boldsymbol{X}) \right] \quad \text{and} \quad \mathcal{U}(\Pb) = \sum_{j = 1}^8 \E\left[ \{d_u(\boldsymbol{X}) = j\} \theta_{u, j}(\boldsymbol{X})\right].
\end{align*}
For each $y,a,z \in \{0,1\}$, define 
\[
\psi_{y a . z}(O ; \Pb)=\frac{\mathds{1}(Z=z)}{\lambda_z(\boldsymbol{X})}\left\{\mathds{1}(Y=y, A=a)-\pi_{y a . z}(\boldsymbol{X})\right\}.
\]
Under the assumption that $\boldsymbol{x} \mapsto d_\ell(\boldsymbol{x})$ and $\boldsymbol{x} \mapsto d_u(\boldsymbol{x})$ are known functions, and are unique maximizers and minimizers, respectively, it can be shown that the uncentered influence functions of $\mathcal{L}(\Pb)$ and $\mathcal{U}(\Pb)$ are
\begin{align*}
& \varphi_\ell(O; \Pb, d_\ell) = \sum_{j = 1}^8 \mathds{1}\{d_\ell(\boldsymbol{X}) = j\} \left\{ L_j(O; \Pb) + \theta_{\ell, j} (\boldsymbol{X}) \right\} \\
& \varphi_u(O; \Pb, d_u) = \sum_{j = 1}^8 \mathds{1}\{d_u(\boldsymbol{X}) = j\} \left\{ U_j(O; \Pb) + \theta_{u, j} (\boldsymbol{X}) \right\}
\end{align*}
where $L_j(O; \Pb)$ and $U_j(O; \Pb)$ are obtained by replacing $\pi_{ya.z}(\boldsymbol{X})$ with $\psi_{ya.z}(O; \Pb)$ in $\theta_{\ell,j}(\boldsymbol{X})$ and $\theta_{u,j}(\boldsymbol{X})$ in display~\eqref{eq:lower-bound}, respectively, and omitting the constant $1$ whenever it appears, e.g., 
\[L_1 = \psi_{11.1} + \psi_{00.0}, \  U_1 = - \psi_{01.1} - \psi_{10.0}.\]
Therefore, under minor regularity conditions, an ``infeasible'' estimator $\widetilde{\mathcal{L}} = \Pn \{\varphi_\ell(O; \widehat\Pb, d_\ell)\} $, i.e., requiring complete knowledge of $d_\ell$, would satisfy 
\begin{align*}
\sqrt{n}(\widetilde{\mathcal{L}}  - \mathcal{L}(\Pb)) \indist N(0, \mathrm{Var}_{\mathbb{P}} \{\varphi_\ell(O; \Pb, d_\ell)\})
\end{align*}
as long as $\Pb\{\varphi_\ell(O; \Pb, d_\ell) - \varphi_\ell(O; \widehat\Pb, d_\ell)\} = o_\Pb(n^{-1/2})$. In this regime, $\widetilde{\mathcal{L}}$ would be efficient. 

\subsection{Estimation \& Inference}
In light of the discussion in the section above, we propose estimating the bounds with
\begin{align*}
& \widehat{\mathcal{L}} = \sum_{j = 1}^8 \Pn\left[ \mathds{1}\{\widehat{d}_\ell(\boldsymbol{X}) = j\} \{L_j(O; \widehat\Pb) + \widehat\theta_{\ell, j}(\boldsymbol{X}) \right] = \Pn \{\varphi_\ell(O; \widehat\Pb, \widehat{d}_\ell) \}, \\
& \widehat{\mathcal{U}} = \sum_{j = 1}^8 \Pn\left[ \mathds{1}\{\widehat{d}_u(\boldsymbol{X}) = j\} \{U_j(O; \widehat\Pb) + \widehat\theta_{u, j}(\boldsymbol{X}) \right] = \Pn \{\varphi_u(O; \widehat\Pb, \widehat{d}_u) \}.
\end{align*}
That is, to estimate the non-smooth components of the bounds $\mathcal{L}(\mathbb{P})$ and $\mathcal{U}(\mathbb{P})$, namely the indicators $\mathds{1}\{d_\ell(\boldsymbol{X}) = j\}$ and $\mathds{1}\{d_u(\boldsymbol{X}) = j\}$, we use plug-in estimators $\mathds{1}\{\widehat{d}_\ell(\boldsymbol{X}) = j\}$ and $\mathds{1}\{\widehat{d}_u(\boldsymbol{X}) = j\}$, where $\widehat{d}_{\ell} \in \argmax_{1 \leq j \leq 8} \widehat{\theta}_{\ell, j}$ and $\widehat{d}_{u} \in \argmax_{1 \leq j \leq 8} \widehat{\theta}_{u, j}$. A natural question is then under what conditions, if any, the estimators $\widehat{\mathcal{L}}$ and $\widehat{\mathcal{U}}$ behave, at least asymptotically, like their infeasible counterparts $\widetilde{\mathcal{L}}$ and $\widetilde{\mathcal{U}}$. As shown in the next theorem, a sufficient condition for such oracle behavior is captured by the following ``margin'' condition. This additional assumption controls the probability that the minimum and maximum are near their points of non-differentiability.

\begin{assumption}[Margin condition] \label{ass:margin} There exists
  $\alpha > 0$ such that for any $t \geq 0$,
  \begin{equation} \label{eq:margin-lower}
  \Pb\left[\min_{j \neq d_{\ell}(\boldsymbol{X})}
  \{\theta_{\ell,
        d_{\ell}(\boldsymbol{X})}(\boldsymbol{X}) -
      \theta_{\ell,j}(\boldsymbol{X})\} \leq t\right ] \lesssim
    t^{\alpha},
    \end{equation}
    and
    \begin{equation} \label{eq:margin-upper}
    \Pb\left[\min_{j \neq d_{u}(\boldsymbol{X})}
  \{
      \theta_{u,j}(\boldsymbol{X}) - \theta_{u,
        d_{u}(\boldsymbol{X})}(\boldsymbol{X})\} \leq t\right ] \lesssim
    t^{\alpha}.
    \end{equation}
\end{assumption}

The margin condition in Assumption~\ref{ass:margin} is very similar to conditions that have been proposed and leveraged in the classification literature~\citep{audibert2007}, as well as in dynamic treatment regimes~\citep{luedtke2016} and other instrumental variable problems~\citep{kennedy2020b}. In words, condition~\eqref{eq:margin-lower} says that with high probability, $\theta_{\ell, d_{\ell}}$ is separated from non-maximal lower bound values $\theta_{\ell, j}$. Similarly, condition~\eqref{eq:margin-upper} limits how close non-minimal upper bound values $\theta_{u, j}$ are to the actual minimum $\theta_{u, d_u}$. If $\min_{j \neq d_{\ell}(\boldsymbol{X})}
  \{\theta_{\ell,
        d_{\ell}(\boldsymbol{X})}(\boldsymbol{X}) -
      \theta_{\ell,j}(\boldsymbol{X})\}$, for instance, has bounded density near zero, then~\eqref{eq:margin-lower} will hold with $\alpha = 1$.
This is a relatively weak requirement which we expect to hold in many cases. Under Assumption \ref{ass:margin}, we are able to derive sufficient conditions such that $\widehat{\mathcal{L}}$ and $\widehat{\mathcal{U}}$ are $\sqrt{n}$-consistent and asymptotically normal.

\begin{theorem}\label{thm:convergence_margin}
Suppose that the nuisance functions $\widehat\pi_{ya.z}$ and $\widehat\lambda_z$ are estimated from a separate independent sample.
Moreover, suppose that Assumption \ref{ass:margin} holds, $\Pb\left(\epsilon \leq \widehat\lambda_1(\boldsymbol{X}) \leq 1 - \epsilon\right) = 1$, for some $\epsilon > 0$, $\left\|\widehat\lambda_1 - \lambda_1\right\| = o_\Pb(1)$, and $\max_{y, a, z \in \{0, 1\}} \|\widehat\pi_{ya.z} - \pi_{ya.z}\| = o_\Pb(1)$.
Then, we have
\begin{align*}
\widehat{\mathcal{L}} - \mathcal{L} & = (\Pn - \Pb)\varphi_\ell(O; \Pb, d_\ell) \\
& \hphantom{=} + O_\Pb\left(\left\lVert \widehat{\lambda}_1 - \lambda_1 \right\rVert
    \cdot \max_{y,a,z \in \{0,1\}} \left\lVert \widehat{\pi}_{ya.z} -
      \pi_{ya.z}\right\rVert + \max_{1 \leq j \leq 8}\left\lVert
      \widehat{\theta}_{\ell, j} - \theta_{\ell, j} \right\rVert_{\infty}^{1 +
      \alpha} \right) + o_\Pb(n^{-1/2}),
\end{align*}
and 
\begin{align*}
\widehat{\mathcal{U}} - \mathcal{U} & = (\Pn - \Pb)\varphi_u(O; \Pb, d_u) \\
& \hphantom{=} + O_\Pb\left(\left\lVert \widehat{\lambda}_1 - \lambda_1 \right\rVert
    \cdot \max_{y,a,z \in \{0,1\}} \left\lVert \widehat{\pi}_{ya.z} -
      \pi_{ya.z}\right\rVert + \max_{1 \leq j \leq 8}\left\lVert
      \widehat{\theta}_{u, j} - \theta_{u, j} \right\rVert_{\infty}^{1 +
      \alpha} \right) + o_\Pb(n^{-1/2}).
\end{align*}
\end{theorem}
In this result, we assume the nuisance functions are estimated on separate independent data for simplicity: in practice with one random sample, one can split the sample and use cross-fitting to achieve the same asymptotic behavior \citep{zheng2010asymptotic, chernozhukov2018double}. Sample splitting allows us to avoid complicated empirical process conditions \citep{chernozhukov2016double, kennedy2020b} and derive bounds on the conditional bias of the estimator in terms of the convergence rate of the nuisance functions. 

Note that in Theorem \ref{thm:dr-smooth} we do not require the individual nuisance functions to converge at $\sqrt{n}$-rates. The conditions are on the product of convergence rates, or rates raised to a power greater than 1. This is a key advantage of a robust estimator: after we correct for the first-order bias, the remaining bias only involves higher-order terms and is much smaller. The result of Theorem \ref{thm:convergence_margin} establishes that $\widehat{\mathcal{L}}$, for instance, is $\sqrt{n}$-consistent as long as 
\begin{align*}
\left\lVert \widehat{\lambda}_1 - \lambda_1 \right\rVert
    \cdot \max_{y,a,z \in \{0,1\}} \left\lVert \widehat{\pi}_{ya.z} -
      \pi_{ya.z}\right\rVert + \max_{1 \leq j \leq 8}\left\lVert
      \widehat{\theta}_{\ell, j} - \theta_{\ell, j} \right\rVert_{\infty}^{1 +
      \alpha} = o_\Pb(n^{-1/2})
\end{align*}
The first product-bias term results from estimation of $\theta_{\ell, j}$ with $L_j(O; \widehat{\mathbb{P}}) + \widehat{\theta}_{\ell, j}$. The second term denotes a bound on the bias arising plug-in estimation of the indicators $\mathds{1}\{d_\ell(\bx) = j\}$ and depends on the exponent $\alpha$ from Assumption \ref{ass:margin}. For example, as long as the quantity $\min_{j \neq d_\ell(\boldsymbol{X})} \{\theta_{\ell, d_\ell(\boldsymbol{X})} - \theta_{\ell, j}(\boldsymbol{X})\}$ has a bounded density near zero, then Assumption \ref{ass:margin} holds with $\alpha = 1$. In this case, $\max_{y,a,z \in \{0,1\}} \left\lVert \widehat{\pi}_{ya.z} -
      \pi_{ya.z}\right\rVert_\infty = o_\Pb(n^{-1/4})$ is typically sufficient to imply
\begin{align*}
\max_{1 \leq j \leq 8}\left\lVert
      \widehat{\theta}_{\ell, j} - \theta_{\ell, j} \right\rVert_{\infty}^{1 +
      \alpha} = o_\Pb(n^{-1/2}).
\end{align*}
 Finally, Theorem \ref{thm:convergence_margin} outlines sufficient conditions for the asymptotic normality of $\widehat{\mathcal{L}}$ and $\widehat{\mathcal{U}}$ so that Wald-type confidence intervals are straightforward to compute. In particular, under the conditions for $\sqrt{n}$-consistency, we can construct an asymptotically valid $100(1 - \delta)\%$ confidence interval for the ATE with
 \[\bigg(\widehat{\mathcal{L}}-z_{1-\delta / 2}\sqrt{\widehat{V}/n} , \ \widehat{\mathcal{U}}+z_{1-\delta / 2}\sqrt{\widehat{W} / n}\bigg),
\]
where $\widehat{V} = \mathbb{P}_n\left[\left\{\varphi_{\ell}(O; \widehat\Pb, \widehat{d}_\ell) - \widehat{\mathcal{L}}\right\}^2\right]$, $\widehat{V} = \mathbb{P}_n\left[\left\{\varphi_{u}(O; \widehat\Pb, \widehat{d}_u) - \widehat{\mathcal{U}}\right\}^2\right]$, and where $z_{\beta}$ is the $\beta$-th quantile of the standard normal distribution.
Moreover, the procedure proposed in \cite{imbens2004confidence} can be used to conduct more precise inferences (see also Theorem 3 in \cite{jiang2018using}).

 We highlight that marginally randomized instruments present a special case of interest. In this case, $Z \ind \boldsymbol{X} $ holds by randomization of treatment assignment $Z$ and $\lambda_{z}(\boldsymbol{X}) $ is equal to a known constant $\lambda_z$ by design. Here, the requirement for the convergence rate is reduced to $\max_{1 \leq j \leq 8}\left\lVert
      \widehat{\theta}_{\ell, j} - \theta_{\ell, j} \right\rVert_{\infty}=o_{\Pb}\left(n^{-\frac{1}{2(1 + \alpha)}}\right)$.

\section{Targeting Smooth Approximations} 
\label{sec:smooth}

 As an alternative to direct estimation of the bounds under a margin condition, our second proposal is to instead target approximations to the bounds, $\mathcal{L}_g(\Pb) = \mathbb{E}_{\mathbb{P}}(g(\boldsymbol{\theta}_{\ell}(\boldsymbol{X})))$ and $\mathcal{U}_h(\Pb) = \mathbb{E}_{\mathbb{P}}(h(\boldsymbol{\theta}_{u}(\boldsymbol{X})))$, where $g, h: [-1,1]^8 \to \mathbb{R}$ are sufficiently smooth approximate pointwise maximum and minimum functions, respectively.

\subsection{Approximation Based on the LSE Function}

While other approximations are possible, we will focus on the log-sum-exp (LSE) function as an approximate maximum. Namely, for any $t > 0$, define $g_t : \mathbb{R}^k \to \mathbb{R}$ via
\begin{equation}\label{eq:LSE}
g_t(\boldsymbol{v}) = \frac{1}{t}\log \left(\sum_{j=1}^k e^{t v_j} \right), \text{ for } \boldsymbol{v} \in \mathbb{R}^k.
\end{equation}
 The LSE is a smooth convex function \citep{boyd2004convex}, with gradient and Hessian given by
\[
\nabla g_t(\boldsymbol{v}) = \frac{\boldsymbol{z}}{\mathbf{1}^T \boldsymbol{z}}, \quad \nabla^2 g_t(\boldsymbol{v}) = \frac{t}{\left(\mathbf{1}^T \boldsymbol{z}\right)^2}\left[\left(\mathbf{1}^T \boldsymbol{z}\right) \operatorname{diag}(\boldsymbol{z})-\boldsymbol{z} \boldsymbol{z}^T\right],
\]
where $\boldsymbol{z} = (e^{t v_1},\dots,e^{t v_k})$. 
Importantly, for our purposes, the following inequality shows that LSE approximates the pointwise maximum function:
\[
\max \left\{v_1, \ldots, v_k\right\}<g_t(\boldsymbol{v}) \leq \max \left\{v_1, \ldots, v_k\right\}+\frac{\log k}{t}.
\]
Increasing the tuning parameter $t$ yields smaller approximation error. On the other hand, the Hessian matrix of $g_t$ also depends critically on $t$; as we will see in the following discussions, as $t$ increases, the operator norm $\left\lVert\nabla^2 g_t(\boldsymbol{v})\right\rVert_{\mathrm{op}}$ increases, which may induce larger estimation error. 

In the context of our problem, we replace the maximum function in $\E_{\mathbb{P}}\left(\max_{1\leq j \leq 8} \theta_{\ell,j}(\boldsymbol{X})\right)$ with the LSE function and focus on estimating the smooth functional $\E_{\mathbb{P}}(g_t( \boldsymbol{\theta}_{\ell}(\boldsymbol{X})))$.
By the approximation property, $\E_{\mathbb{P}}(g_t( \boldsymbol{\theta}_{\ell}(\boldsymbol{X})))$ satisfies
\[
\E\left(\max_{1\leq j \leq 8} \theta_{\ell,j}(\boldsymbol{X})\right) \leq \E(g_t( \boldsymbol{\theta}_{\ell}(\boldsymbol{X}))) \leq \E\left(\max_{1\leq j \leq 8} \theta_{\ell,j}(\boldsymbol{X})\right) + \frac{\log 8}{t}.
\]
We can similarly define a smooth approximation for pointwise minimum function as $h_t = g_{-t}$ and estimate the smooth functional $\E_{\mathbb{P}}(h_t( \boldsymbol{\theta}_{u}(\boldsymbol{X})))$ for the upper bound $\E_{\mathbb{P}}\left(\min_{1\leq j \leq 8} \theta_{u,j}(\boldsymbol{X})\right)$.
The smooth approximation for the minimum function satisfies
\[
\E\left(\min_{1\leq j \leq 8} \theta_{u,j}(\boldsymbol{X})\right) -\frac{\log 8}{t} \leq \E(h_t( \boldsymbol{\theta}_{u}(\boldsymbol{X}))) \leq \E\left(\min_{1\leq j \leq 8} \theta_{u,j}(\boldsymbol{X})\right).
\]
In the remaining part of this section, we develop efficiency theory for these smooth functional approximations of the Balke-Pearl bounds, and propose robust and efficient estimators.

\subsection{Efficiency Theory}

In the following theorem, we present the nonparametric efficient influence function for functionals of the form $\mathcal{L}_g(\Pb) = \mathbb{E}_{\mathbb{P}}(g(\boldsymbol{\theta}_{\ell}(\boldsymbol{X})))$ and $\mathcal{U}_h(\Pb) = \mathbb{E}_{\mathbb{P}}(h(\boldsymbol{\theta}_{u}(\boldsymbol{X})))$, where $g,h: [-1,1]^8 \mapsto \R$ are smooth functions with continuous first-order derivatives.


\begin{theorem}\label{thm:IF-smooth}
The nonparametric influence function of $\mathcal{L}_g(\Pb)$ is 
\[
\dot{\mathcal{L}}_g(O ; \Pb)=g\left(\boldsymbol{\theta}_{\ell}(\boldsymbol{X})\right)-\mathcal{L}_g(\Pb)+\sum_{j=1}^8 \frac{\partial g\left(\boldsymbol{\theta}_{\ell}(\boldsymbol{X})\right)}{\partial \theta_{\ell, j}(\boldsymbol{X})} L_j(O ; \Pb),
\]
where $L_j(O ; \Pb)$ is defined in Section~\ref{sec:infeasible}. Similarly, the nonparametric influence function of $\mathcal{U}_h(\Pb)$ is 
\[
\dot{\mathcal{U}}_h(O ; \Pb)=h\left(\boldsymbol{\theta}_{u}(\boldsymbol{X})\right)-\mathcal{U}_h(\Pb)+\sum_{j=1}^8 \frac{\partial h\left(\boldsymbol{\theta}_{u}(\boldsymbol{X})\right)}{\partial \theta_{u, j}(\boldsymbol{X})} U_j(O ; \Pb),
\]
again with $U_j(O ; \Pb)$ defined in Section~\ref{sec:infeasible}.
\end{theorem}
\noindent Theorem \ref{thm:IF-smooth} implies that there are two terms contributing to the influence functions (and hence our proposed estimators to follow) of the smooth functionals $\mathcal{L}_g(\Pb)$ and $\mathcal{U}_h(\Pb)$. The first term for the lower bound, $g\left(\boldsymbol{\theta}_{\ell}(\boldsymbol{X})\right)-\mathcal{L}_g(\Pb)$, is augmented with the second term, $\sum_{j=1}^8 \frac{\partial g\left(\boldsymbol{\theta}_{\ell}(\boldsymbol{X})\right)}{\partial \theta_{\ell, j}(\boldsymbol{X})} L_j(O ; \Pb)$, which effectively reduces bias that results from estimating the unknown functions $\pi_{ya.z}$ and $\lambda_z$. 

After characterizing the influence functions for $\mathcal{L}_g(\Pb)$ and $\mathcal{U}_h(\Pb)$, we can use them to correct for the first-order bias in the von Mises expansion and arrive at robust estimators, which allows us to perform estimation and inference efficiently.

\subsection{Estimation \& Inference}
Next, we propose and analyze a robust estimator for the smooth lower bound functional, $\mathcal{L}_g(\Pb)$. Similar results hold for $\mathcal{U}_h(\Pb)$ using the same arguments. As in Section~\ref{sec:binary}, we assume that we train models for the nuisance functions $\pi_{y a . z}(\boldsymbol{X}), \lambda_z(\boldsymbol{X})$ based on a separate independent sample $D^n$. The robust estimator is defined as:
\[
\widehat{\mathcal{L}}_g=\mathcal{L}_g(\widehat{\Pb})+\mathbb{P}_n\left[\dot{\mathcal{L}}_g(O ; \widehat{\Pb})\right]=\mathbb{P}_n\left[g\left(\widehat{\boldsymbol{\theta}}_{\ell}(\boldsymbol{X})\right)+\sum_{j=1}^8 \frac{\partial g\left(\widehat{\boldsymbol{\theta}}_{\ell}(\boldsymbol{X})\right)}{\partial \widehat{\theta}_{\ell, j}(\boldsymbol{X})} L_j(O ; \widehat{\Pb})\right]
\]
The following theorem characterizes the conditional bias (given the training data) of the robust estimator $\widehat{\mathcal{L}}_g$ and establishes its asymptotic normality under additional conditions.
\begin{theorem}\label{thm:dr-smooth}
Suppose $g: \R^8 \mapsto \R $ is twice continuously differentiable, such that $\|\nabla g(\boldsymbol{\theta})\|_{\infty} \leq C_1$ and $\|\nabla^2 g(\boldsymbol{\theta})\|_{\mathrm{op}} \leq C_2$ uniformly over $\boldsymbol{\theta}$, and the nuisance functions $\widehat{\pi}_{ya.z}, \widehat{\lambda}_z$ are estimated from a separate independent sample, $D^n$. Moreover, suppose there exists positive constant $\epsilon$ such that 
\[
\Pb\left(\epsilon \leq \widehat{\lambda}_1(\boldsymbol{X}) \leq 1-\epsilon\right)=1.
\]
Then the conditional bias of the robust estimator $\widehat{\mathcal{L}}_g$ can be bounded as
\[
\begin{aligned}
    &\, \left|\E[\widehat{\mathcal{L}}_g \mid D^n] - \mathcal{L}_g(\Pb) \right| \\
    \lesssim &\, \left(\max _{y, a, z \in\{0,1\}}\left\|\widehat{\pi}_{y a . z}-\pi_{y a . z}\right\|\right)\left(C_1\left\|\widehat{\lambda}_1-\lambda_1\right\|+C_2 \max _{y, a, z \in\{0,1\}}\left\|\widehat{\pi}_{y a . z}-\pi_{y a . z}\right\|\right).
\end{aligned}
\]
Let $f(O)=\dot{\mathcal{L}}_g (O;\Pb) + \mathcal{L}_g(\Pb)$ be the non-centered influence function. If we further assume $\left\|\widehat{f}-f\right\|=o_\Pb(1)$ and the nuisance estimators satisfy
\[
\left\|\widehat{\lambda}_1-\lambda_1\right\|\left(\max _{y, a, z \in\{0,1\}}\left\|\widehat{\pi}_{y a . z}-\pi_{y a . z}\right\|\right) = o_\Pb(n^{-1/2}),
\]
\[
\max _{y, a, z \in\{0,1\}}\left\|\widehat{\pi}_{y a . z}-\pi_{y a . z}\right\|^2 = o_\Pb(n^{-1/2}),
\]
then we have
\[
\widehat{\mathcal{L}}_g - \mathcal{L}_g(\Pb) = \mathbb{P}_n\left[\dot{\mathcal{L}}_g(O ; \Pb)\right] + o_\Pb(n^{-1/2}),
\]
implying the robust estimator is $\sqrt{n}$-consistent and achieves the nonparametric efficiency bound. 
\end{theorem}
 Similar to the analysis following Theorem~\ref{thm:convergence_margin}, Theorem \ref{thm:dr-smooth} implies that if the nuisance estimators satisfy $\left\|\widehat{\lambda}_1-\lambda_1\right\|=O_{\Pb}\left(n^{-1 / 4}\right)$ and $\max _{y, a, z \in\{0,1\}}\left\|\widehat{\pi}_{y a . z}-\pi_{y a . z}\right\|=o_{\Pb}\left(n^{-1 / 4}\right)$, the robust estimator will be $\sqrt{n}$-consistent and achieve the efficiency bound. Comparing the direct estimator from Section~\ref{sec:binary} with the approximation-based estimator here, when Assumption \ref{ass:margin} holds with $\alpha = 1$ then the requirements on the bias of Theorem \ref{thm:convergence_margin} are essentially the same as those of Theorem \ref{thm:dr-smooth}---this is because estimation in $L_\infty$ typically yields the same convergence rate as estimation in $L_2$ up to a $\log$ factor \citep{tsybakov2009}.

 \begin{remark}
 It is natural to consider, assuming that the margin condition described in Assumption~\ref{ass:margin} holds, (a) whether and how the behavior of the smooth approximation estimator $\widehat{\mathcal{L}}_g$ might improve, (b) whether and how to choose a sequence of tuning parameters $t_n$ in the LSE function (or some other smooth approximation to the pointwise maximum) to minimize the convergence rate and mean square error of $\widehat{\mathcal{L}}_{g_{t_n}}$ with respect to $\mathcal{L}(\mathbb{P})$, and (c) if we can describe minimax optimal estimators of $\mathcal{L}(\mathbb{P})$ under a large class of models (e.g., nuisance functions belonging to H{\"o}lder smooth classes). Careful analysis in a simplified setting (omitted here) reveals that while the approximation error $|\mathcal{L}_{g_{t_n}}(\mathbb{P}) - \mathcal{L}(\mathbb{P})|$ improves from $O(\frac{1}{t_n})$ to $O(\frac{1}{t_n^{1 + \alpha}})$ under a margin condition, the direct estimator dominates the smooth approximation even for optimally chosen data-adaptive tuning parameters $t_n$. Indeed, we conjecture that the direct estimator $\widehat{\mathcal{L}}$, modulo higher order influence function corrections \citep{robins2008higher, robins2009quadratic, robins2017higher}, is rate-optimal (up to log factors) when Assumption~\ref{ass:margin} holds. We leave precise minimax optimal rate characterization---for non-smooth functionals similar to $\mathcal{L}(\mathbb{P})$ and $\mathcal{U}(\mathbb{P})$ under a margin condition akin to Assumption~\ref{ass:margin}---to be investigated in future research. 
 \end{remark}


\subsection{Wald-type Confidence Interval for the ATE}\label{sec:approxinterval}
Combining Theorem \ref{thm:BP} with the smooth LSE-based approximation we have
\[
\mathcal{L}_{g_t}(\Pb) -\frac{\log 8}{t} \leq \mathbb{E}[Y(a = 1) - Y(a = 0)] \leq \mathcal{U}_{h_t}(\Pb) + \frac{\log 8}{t}.
\]
Thus the interval
\[
\left(\widehat{\mathcal{L}}_{g_t}-\frac{\log 8}{t}-z_{1-\delta / 2} \sqrt{\widehat{V}_t / n }, \ \widehat{\mathcal{U}}_{h_t}+\frac{\log 8}{t}+z_{1-\delta / 2} \sqrt{\widehat{W}_t / n }\right)
\]
is an asymptotically valid (though potentially conservative) $100(1-\delta)\%$ Wald-type confidence interval for the ATE, where $\widehat{V}_t, \widehat{W}_t$ are plug-in estimators of the nonparametric efficiency bounds $\operatorname{Var}_{\Pb}\left(\dot{\mathcal{L}}_{g_t}(O ; \Pb)\right)$ and $\operatorname{Var}_{\Pb}\left(\dot{\mathcal{U}}_{h_t}(O ; \Pb)\right)$, respectively. The choice of the tuning parameter $t$ requires balancing the smooth approximation error $\log 8/t$ and the conditional bias term $\E[\widehat{\mathcal{L}}_g \mid D^n] - \mathcal{L}_g(\Pb)$. Note that $\|\nabla^2 g_t(\boldsymbol{\theta})\|_{\mathrm{op}} \leq t$ and Theorem \ref{thm:dr-smooth} implies the conditional bias increases as we select a larger $t$. Hence a smaller $t$ is preferred to reduce the conditional bias. Yet reducing the approximation error $\log 8/t$ requires a larger $t$. How to choose $t$ in a data-driven fashion to arrive at the shortest confidence interval remains an open problem left for future investigation.

\section{Direct Estimation with Continuous Outcomes} 
\label{sec:continuous}



Thus far, we have only focused on a binary outcome $Y \in \{0,1\}$. In this section, we will extend our approaches from the binary outcome case to construct valid bounds on the ATE for continuous outcomes. In fact, we will assume only that the outcome $Y$ is bounded; without loss of generality, we may assume $Y \in [0,1]$ (otherwise one can always rescale the outcome). As observed but not studied in \cite{balke1997bounds}, the idea is to replace the binary outcome in previous analyses with the indicator $\mathds{1}(Y\leq t)$. Specifically, for each $t \in [0,1]$, $\mathds{1}(Y\leq t)$ is a binary outcome for which we can apply the bounds proposed in Section \ref{sec:BP-bounds} to obtain pointwise bounds on the difference in the distribution functions of potential outcomes $\Pb(Y(a) \leq t)$. We proceed by integrating the tail probabilities $\Pb(Y(a)>t)$ to bound the mean of potential outcomes $\E(Y(a))$, and finally arrive at bounds on the ATE, $\E(Y(a = 1) - Y(a = 0))$. 

For any $a,z \in \{0,1\}$ and $t \in \R$, define the probabilities $\pi_{1 a . z}(t, \boldsymbol{X})=\Pb(Y \leq t, A=a \mid \boldsymbol{X}, Z=z)$ and $\pi_{0 a . z}(t, \boldsymbol{X})=\Pb(Y > t, A=a \mid \boldsymbol{X}, Z=z)$. 
Next, we define functions exactly as in \eqref{eq:lower-bound} and \eqref{eq:upper-bound}, replacing with $\pi_{ya.z}(\boldsymbol{X})$ with $\pi_{ya.z}(t,\boldsymbol{X})$ (note that $\theta_{u,j}$, $\theta_{\ell,j}$, $\gamma_{\ell}$, and $\gamma_{u}$ are then also functions of $t$ and $\boldsymbol{X}$). Assuming $Z$ is a valid instrument given covariates $\boldsymbol{X}$, for each $t \in \R$ we view $\mathds{1}(Y \leq t)$ as the binary outcome and apply the ``Balke-Pearl'' bounds to this new outcome. By Theorem \ref{thm:BP}, we obtain the following
bounds on the difference in distribution functions of the individual potential outcomes:
\[
\gamma_{\ell}(t,\boldsymbol{X}) \leq \mathbb{P}[Y(a = 1) \leq t \mid \boldsymbol{X}] - \mathbb{P}[Y(a = 0) \leq t \mid \boldsymbol{X}] \leq \gamma_u(t, \boldsymbol{X}).
\]
Integrating with respect to $\boldsymbol{X}$ and noting that $\mathbb{E}(Y(a))=\int_0^1 \Pb[Y(a)>t] dt$, we arrive at the following valid (though not necessarily tight) bounds on the ATE:
\begin{equation}\label{eq:integrated-bounds}
- \int_0^1\mathbb{E} \left(\gamma_u(t, \boldsymbol{X})\right) d t \leq \mathbb{E}(Y(a = 1)-Y(a =0)) \leq - \int_0^1\mathbb{E} \left(\gamma_\ell(t, \boldsymbol{X})\right) d t.
\end{equation}
One natural approach would be to estimate the functions $t \mapsto \mathbb{E} \left(\gamma_u(t, \boldsymbol{X})\right)$ and $t \mapsto \mathbb{E} \left(\gamma_\ell(t, \boldsymbol{X})\right)$ by the proposals of Section \ref{sec:binary} or Section \ref{sec:smooth}, over a grid of values $t \in [0,1]$, and use these estimates to approximate the integrated bounds given in \eqref{eq:integrated-bounds}. Alternatively, one could use Monte Carlo methods to approximate the integral. Either approach would require computing a binary-outcome estimator at many different inputs $t$, which could be quite computationally intensive. The following theorem characterizes a looser bound that can be computed more efficiently.

\begin{theorem}\label{thm:continuous-cdf}
Let $W \sim \mathrm{Uniform}(0,1)$ be independent of $O = (\boldsymbol{X}, Z, A, Y)$. Then
\[
\int_0^1\mathbb{E} \left(\gamma_u(t, \boldsymbol{X})\right) d t \leq \mathbb{E}\left[\min _{1 \leq j \leq 8} \widetilde{\theta}_{u, j}(\boldsymbol{X})\right],
\]
\[
\int_0^1\mathbb{E} \left(\gamma_\ell(t, \boldsymbol{X})\right) d t \geq \mathbb{E}\left[\max _{1 \leq j \leq 8} \widetilde{\theta}_{\ell, j}(\boldsymbol{X})\right],
\]
where we define $\widetilde{\pi}_{1 a . z}(\boldsymbol{X})=\Pb(Y \leq W, A=a \mid \boldsymbol{X}, Z=z), \widetilde{\pi}_{0 a . z}(\boldsymbol{X})=\Pb(Y > W, A=a \mid \boldsymbol{X}, Z=z)$, and where $\widetilde{\theta}_{\ell,j}(\boldsymbol{X})$ and $\widetilde{\theta}_{u,j}(\boldsymbol{X})$ are defined exactly as in equations \eqref{eq:lower-bound} and \eqref{eq:upper-bound}, replacing $\pi$ with $\widetilde{\pi}$.
\end{theorem}
Theorem \ref{thm:continuous-cdf} together with \eqref{eq:integrated-bounds} implies the following valid bounds on the ATE:
\begin{equation}\label{eq:continuous-bounds}
-\mathbb{E}\left[\min _{1 \leq j \leq 8} \widetilde{\theta}_{u, j}(\boldsymbol{X})\right] \leq \E[Y(1)-Y(0)] \leq -\mathbb{E}\left[\max _{1 \leq j \leq 8} \widetilde{\theta}_{\ell, j}(\boldsymbol{X})\right].
\end{equation}
Note that the bounds in \eqref{eq:continuous-bounds} are of the same form as those in Theorem \ref{thm:BP}, but with a new binary outcome $\mathds{1}(Y \leq W)$. By introducing an independent uniform random variable, we obtain a looser bound, but one that lends itself to more computationally feasible estimation. To operationalize these bounds, we can simulate $n$ i.i.d. $\mathrm{Uniform}(0,1)$ random variables $W_1 ,\dots, W_n$, independent of the data, and construct the augmented observation unit $\widetilde{O} = (\boldsymbol{X},Z,A,Y,W)$. We then estimate the quantities in the bounds exactly as in the binary case using the methods proposed in Section \ref{sec:binary} or Section \ref{sec:smooth}, with the new binary outcome being $\mathds{1}(Y \leq W)$. Note that the exact value of the final bounds will depend on the realization of $W_1 ,\dots, W_n$. To reduce this variability and regain some efficiency, we may repeat this process $m$ times and average the resulting estimates of the bounds---see Proposition \ref{prop:simp} in Appendix \ref{sec:appendix-continuous} for an illustration of this behavior in a simplified setting.






\section{Simulation Study}
\label{sec:simulation}

In order to assess the performance of the proposed estimators, we conducted a small simulation study. In particular, we consider a simplified scenario where Assumption~\ref{ass:margin} holds with margin parameter $\alpha = 1$, and demonstrate for different sample sizes $n$ how mean squared errors of the proposed lower bound estimators $\widehat{\mathcal{L}}$ and $\widehat{\mathcal{L}}_{g_t}$ compare to that of a plug-in estimator. We find that, as anticipated by our theoretical results, our proposed estimators generally have smaller error, and achieve parametric levels of error when the root-mean-square error of nuisance estimates is less than $O(n^{-1/4})$. Details and full results are included in Appendix~\ref{sec:appendix-simulation}.

\section{Data Analysis}
\label{sec:data}

To illustrate the proposed methods on real data, we aimed to estimate the effect of higher education on wages later in life using a subset of data from the National Longitudinal Survey of Young Men \citep{card1995}, as previously analyzed, for instance, in \citet{tan2006, wang2018}. For ease of comparison, we replicate the setup of \citet{wang2018}, and use the $n = 3010$ participants (males between 14 and 24 years in 1966) for whom survey responses for education and wage in 1976 were available. In particular, the instrument $Z$ is taken to be an indicator of proximity to a 4-year college, the exposure $A$ an indicator of education beyond high school, and the outcome $Y$ an indicator of wages in 1976 exceeding the median value. In addition to $(Z, A, Y)$, we adjusted for a collection of baseline covariates $\boldsymbol{X}$: age, parental educational attainment, indicators for living in the south and a metropolitan area, race, intelligence quotient scores, as well as variable-specific missingness indicators. As in \citet{card1995} and \citet{wang2018}, mean imputation was used for missing covariates, and, due to non-representativeness of the sample, the observations were reweighted using sampling weights (results from unweighted analyses were qualitatively similar).

\citet{wang2018} consider the binary instrumental variable problem (i.e., $Z, A, Y \in \{0,1\}$) as we do, and propose several estimators that respect that the CATE and ATE must be bounded between $-1$ and 1. Note that, in addition to the assumptions adopted here which underpin the Balke-Pearl bounds, these authors introduced a structural assumption---asserting no additive interaction between unmeasured confounders and either (i) the instrument, on the treatment, or (ii) the treatment, on the outcome---which results in point identification of the CATE and thus the population ATE. That is, their methods rely on strictly stronger assumptions than ours. 

We computed estimates of the covariate-adjusted bounds $\mathcal{L}(\Pb)$ and $\mathcal{U}(\Pb)$ using the direct estimators $\widehat{\mathcal{L}}$ and $\widehat{\mathcal{U}}$ developed in Section~\ref{sec:binary}, assuming the margin condition (i.e., Assumption~\ref{ass:margin}). Specifically, we used separate ensembles based on fits from \texttt{glm}, \texttt{rpart}, \texttt{ranger}, and \texttt{polymars} using the \texttt{SuperLearner} package in \texttt{R} \citep{polley2019} to obtain $\widehat{\lambda}_1(\boldsymbol{X})$ and $\{\widehat{\pi}_{ya.z}(\boldsymbol{X}): y, a, z \in \{0,1\}\}$. We also computed LSE-based estimators $\widehat{\mathcal{L}}_{g_t}$ and $\widehat{\mathcal{U}}_{h_t}$ based on the same nuisance function estimators, as described in Section~\ref{sec:smooth}, with the \textit{ad hoc} tuning parameter choice $t = 100 \, n^{1/4}$. For all estimators, we used cross-fitting with 5 folds.

Summarizing the results, the point estimates for the covariate-adjusted Balke-Pearl bound estimators $\widehat{\mathcal{L}}$ and $\widehat{\mathcal{U}}$ were $(-0.405, 0.546)$, and the resulting Wald-based 95\% confidence interval for the ATE was $(-0.454, 0.597)$. These results indicate a wide range of possible effects compatible with the observed data, with the width perhaps demonstrating the potentially strong effects of unmeasured confounders. Nevertheless, quite surprisingly, compared to the bounds reported in \citet{wang2018}, the confidence interval width is narrower and the estimated upper bound is more informative. For instance, the point estimates and 95\% confidence intervals for their proposed ``$g$-estimator'' and ``bounded multiply robust estimator'' were 0.079 $(-0.355, 1.000)$ and 0.344 $(-0.373, 0.938)$, respectively. This discrepancy may arise as their estimators rely on unstable weights which do not feed into our approach (i.e., the nonparametric efficiency bounds for the Balke-Pearl bound functionals are much less than that for the functional targetted by these authors), or because the constraints implied by the binary instrumental variable model are not explicitly incorporated in their approaches (cf. Proposition 1 in \citet{wang2018}).

The approximation-based estimators $\widehat{\mathcal{L}}_{g_t}$ and $\widehat{\mathcal{U}}_{h_t}$ were nearly identical, $(-0.406, 0.544)$, and the associated 95\% confidence interval constructed as in Section~\ref{sec:approxinterval} was $(-0.506, 0.646)$: the extra width is due to the inclusion of the (potentially conservative) error bound $\log{(8)} / t$.

\section{Discussion}\label{sec:discussion}
In this paper, we proposed estimators of nonparametric bounds on the average treatment effect using an instrumental variable, avoiding strong structural or parametric assumptions typically used for point identification. We extended the classic approach of~\citet{balke1997bounds} by incorporating baseline covariate information to (i) control for measured confounders to render the instrument valid, and/or (ii) construct narrower and hence more informative bounds. We also proposed a concrete extension to bound the ATE for a general bounded outcome. Our estimators are based on influence functions, and as a result are robust and can attain $\sqrt{n}$-consistency, asymptotic normality, and nonparametric efficiency, under nonparametric convergence rates of the component nuisance functions estimation---these rates are attainable under sparsity, smoothness or other structural conditions.

The key difficulty in estimating the covariate-adjusted bound functionals $\mathcal{L}(\Pb)$ and $\mathcal{U}(\Pb)$ is that, as means of non-differentiable functions, they are not pathwise differentiable \citep{bickel1993efficient}. As a result, these parameters do not---at least without further assumptions---have an influence function to facilitate flexible and efficient estimation (e.g., see Section 5.3 of \citet{kennedy2022semiparametric}). To make progress, in Section~\ref{sec:binary}, we first proposed and invoked a margin condition to render the exact bound functionals pathwise differentiable, and proposed influence function-based estimators under this condition. Second, in Section~\ref{sec:smooth}, we presented general valid bounds on the ATE via estimation of pathwise differentiable approximations to $\mathcal{L}(\Pb)$ and $\mathcal{U}(\Pb)$. The second approach has the advantage of not requiring the margin condition, but the resulting bounds suffer from being slightly conservative, and not converging to the true exact bounds at fast rates. The first approach, on the other hand, can achieve parametric rates when the margin condition holds. As a general guideline, we recommend that practitioners use the direct bound estimators proposed in Section~\ref{sec:binary}, deferring to the approximate bounds only when there is serious doubt on Assumption~\ref{ass:margin} arising due to subject matter knowledge.

In our view, it is difficult to imagine concrete scenarios in which Assumption~\ref{ass:margin} would be violated. A similar margin condition is invoked in dynamic treatment regime problems, for example when estimating the mean outcome under the optimal treatment policy in a point treatment setting: $\mathbb{E}[\max\{\E(Y(a = 1) \mid \bX), \E(Y(a = 0) \mid \bX) \} ]$ \citep{luedtke2016}. In that context, one must argue that the CATE is well separated from zero, the value representing no treatment effect. It is plausible that there is a subgroup for which the CATE is exactly zero, which would violate the margin condition for that problem. In our setting, Assumption~\ref{ass:margin} requires separation of the pointwise maximum (minimum) of the Balke-Pearl lower (upper) bound functions $\boldsymbol{\theta}_{\ell}(\boldsymbol{X})$ ($\boldsymbol{\theta}_{u}(\boldsymbol{X})$) from the second largest (smallest) value. These lower and upper bound functions are not readily interpretable like the CATE, and it would seem to necessitate an unlikely confluence of factors to violate the required separation.

The approaches developed in this paper are not inherently tied to the specific bounds we studied: the same ideas extend to many functionals defined as expectations of non-smooth functions, which arise often in problems targetting bounds on causal effects. Indeed, one month after our work was posted on arXiv, \cite{semenova2023adaptive} independently considered estimation of bounds, and functionals theoreof, that can be expressed as the minimum of several nuisance functions. They studied a wide range of bounds obtained from certain optimization problems and proposed estimators under a margin condition similar to Assumption \ref{ass:margin}. The Balke-Pearl bounds arise from a linear programming specification, which is a special case of bounds based on general optimization problems in \cite{semenova2023adaptive}. However, our work specifically characterizes the nonparametric efficient estimator of the Balke-Pearl bounds, as well as its asymptotic bias, in Theorem \ref{thm:convergence_margin}. Further, a direct application of the approach of \citet{semenova2023adaptive} to our setting would require that instrument probabilities $\lambda_1(\boldsymbol{X})$ are known. We also propose estimators based on smooth approximation---not considered in \citet{semenova2023adaptive}---and establish corresponding efficiency theory. We expect that the methods discussed in this paper can be combined with the insights from \citet{semenova2023adaptive} to yield efficient, debiased estimators in more general settings as considered in their paper. 

It is important to mention some limitations of the proposed methods, as well as possible extensions and open problems. First, we have throughout assumed a fixed collection of measured covariates $\boldsymbol{X}$. In Proposition~\ref{prop:width}, we provide a general criterion to justify adjusting for certain covariates, however, it remains to characterize (i) the actual difference in the length of the bounds based on two valid adjustment sets, and (ii) the effect of the adjustment set on the variance of the proposed estimators. Second, we have focused primarily on the setting with instrument, exposure, and outcome all binary variables, since this is the simplest setting for partial identification with closed form bounds given by \citet{balke1997bounds}. In Section~\ref{sec:continuous}, we provided an extension to construct valid ATE bounds for continuous outcomes, though an important open problem is determining the sharpest possible bounds for such general outcomes, beyond the binary case. Along the same lines, it of course will also be of interest to consider multi-valued or continuous instruments and exposures. In these cases, there is generally no longer the same closed form solution for the theoretical bounds, and alternative optimization specifications and/or symbolic bounds might be incorporated \citep{sachs2022general, zhang2021bounding, duarte2023automated}. When closed-form solutions are available (see examples in \cite{sachs2022general}), similar direct estimation and targeting smooth approximation strategies can be applied to estimate the bounds efficiently. \cite{duarte2023automated} recently proposed a general framework to automatically compute sharp bounds on causal effects under user-specified causal assumptions. Their idea is to formalize the assumptions and observed evidence as polynomial constraints, and estimands of interest as objective functions, following which the bounds are obtained by optimizing the objective. In this regard the Balke-Pearl bounds can be viewed as a special case where the causal assumptions are specified by an IV model. However, the Balke-Pearl bounds admit closed-form solutions, based on which efficiency estimation theory can be derived, while there is no closed-form solutions to the general optimization problem in \cite{duarte2023automated}. Moreover, they only focus on the fully discrete setting since the optimization problem is intractable in greater generality, while the methodology in our work naturally handles continuous covariates (and outcomes using the extension in Section \ref{sec:continuous}). Third, an interesting problem we will explore in future research is to directly and efficiently estimate the conditional bounds $\gamma_{\ell}(\boldsymbol{X})$, $\gamma_u(\boldsymbol{X})$ for the CATE given in Theorem~\ref{thm:BP}, and identify strata for whom we are confident that the CATE is positive or negative. In doing so, one may extend the ideas in \citet{kennedy2020b} to define ``sharpness'' of an instrument---the ability to tightly bound causal effects in certain subgroups---without relying on the typical monotonicity assumption.


\section*{Acknowledgements}
This work was supported by the Patient-Centered Outcomes Research Institute (PCORI Awards ME-2021C1-22355) and the National Library of Medicine, \#1R01LM013361-01A1. All statements in this report, including its findings and conclusions, are solely those of the authors and do not necessarily represent the views of PCORI or its Methodology Committee. 

\clearpage

\section*{References}
\vspace{-1cm}
\bibliographystyle{asa}
\bibliography{bibliography.bib}

\pagebreak

\begin{appendices}

\section{Proofs of Results in Section~\ref{sec:background}}\label{sec:appendix-background}
\subsection{Proof of Theorem~\ref{thm:BP}}
The argument of \citet{balke1997bounds} is valid within levels of $\boldsymbol{X}$, given Assumptions~\ref{ass:consistency}--\ref{ass:ER}. This implies that for each $\boldsymbol{x} \in \mathcal{X}$,
\begin{equation}\label{eq:cond_bound_valid}
    \gamma_{\ell}(\boldsymbol{x}) \leq \mathbb{E}(Y(a = 1) - Y(a = 0) \mid \boldsymbol{X} = \boldsymbol{x}) \leq \gamma_{u}(\boldsymbol{x}),
\end{equation}
and that these bounds are tight. It follows that for any $\epsilon > 0$, there exist joint distributions $\mathbb{P}_{1, \epsilon}^*, \mathbb{P}_{2, \epsilon}^*$ on the full data $(\boldsymbol{X}, Z, A, Y(a = 0), Y(a = 1))$, compatible with the observed data law $\mathbb{P}$ such that 
\begin{equation}\label{eq:cond_bound_tight}
    \mathbb{P}\left(\mathbb{E}_{\mathbb{P}_{1, \epsilon}^*}(Y(a = 1) - Y(a = 0) \mid \boldsymbol{X}) < \gamma_{\ell}(\boldsymbol{X}) + \epsilon\right) = 1,
\end{equation} and
\[\mathbb{P}\left(\mathbb{E}_{\mathbb{P}_{2, \epsilon}^*}(Y(a = 1) - Y(a = 0) \mid \boldsymbol{X}) > \gamma_{u}(\boldsymbol{X}) - \epsilon\right) = 1.\]
By iterated expectations, $\mathbb{E}_{\mathbb{P}}(\gamma_{\ell}(\boldsymbol{X}))$ and $\mathbb{E}_{\mathbb{P}}(\gamma_{u}(\boldsymbol{X}))$ represent valid lower and upper bounds, respectively, on the ATE. To prove that these are also tight, we proceed by contradiction: we focus on the lower bound without loss of generality, and assume there exists $\epsilon > 0$ such that for all full data laws $\mathbb{P}^*$ compatible with $\mathbb{P}$, $\mathbb{E}_{\mathbb{P}^*}(Y(a = 1) - Y(a = 0)) \geq \mathbb{E}_{\mathbb{P}}(\gamma_{\ell}(\boldsymbol{X})) + \epsilon$. But~\eqref{eq:cond_bound_tight} implies $\mathbb{E}_{\mathbb{P}_{1, \epsilon}^*}(Y(a = 1) - Y(a = 0)) < \mathbb{E}_{\mathbb{P}}(\gamma_{\ell}(\boldsymbol{X})) + \epsilon$, which yields a contradiction.

\subsection{Proof of Proposition~\ref{prop:width}}
  The bounds based on $(\boldsymbol{X}, \boldsymbol{G})$ are of the
  same form as those for $\boldsymbol{X}$ alone, but with underlying
  probabilities
  $\pi_{ya.z}^\dagger(\boldsymbol{X}, \boldsymbol{G}) = \mathbb{P}[Y=y, A = a
  \mid \boldsymbol{X}, \boldsymbol{G}, Z = z]$ instead of
  $\pi_{ya.z}(\boldsymbol{X})$. The key observation is that for any
  $y,a,z \in \{0,1\}$,
  \[\pi_{ya.z}(\boldsymbol{X}) =
    \mathbb{E}(\pi_{ya.z}^{\dagger}(\boldsymbol{X}, \boldsymbol{G})
    \mid \boldsymbol{X}, Z = z) =
    \mathbb{E}(\pi_{ya.z}^{\dagger}(\boldsymbol{X}, \boldsymbol{G})
    \mid \boldsymbol{X}),\] where the first equality holds by the
  tower law, and the second holds by the assumption
  $Z \ind \boldsymbol{G} \mid \boldsymbol{X}$. Therefore,
  letting $\theta_{\ell, j}^{\dagger}(\boldsymbol{X}, \boldsymbol{G})$,
  $\theta_{u, j}^{\dagger}(\boldsymbol{X}, \boldsymbol{G})$ be of the
  same for as $\theta_{\ell, j}(\boldsymbol{X})$,
  $\theta_{u, j}(\boldsymbol{X})$, respectively, for $j = 1, \ldots, 8$, but with
  $\pi_{ya.z}^\dagger$ replacing $\pi_{ya.z}$ for all
  $y, a, z \in \{0,1\}$, we have
  \[\mathbb{E}_{\Pb}\left(\max_{1 \leq j \leq 8} \theta_{\ell,
        j}(\boldsymbol{X})\right) = \mathbb{E}_{\Pb}\left(\max_{1 \leq j
        \leq 8} \mathbb{E}_\mathbb{P}(\theta_{\ell, j}^{\dagger}(\boldsymbol{X},
      \boldsymbol{G}) \mid \boldsymbol{X})\right) \leq
    \mathbb{E}_{\Pb}\left(\max_{1 \leq j \leq 8} \theta_{\ell,
        j}^{\dagger}(\boldsymbol{X}, \boldsymbol{G})\right),\] where
  we used a conditional version of Jensen's inequality for the
  pointwise maximum, and iterated expectations. The proof that the
  $(\boldsymbol{X}, \boldsymbol{G})$-upper bound is lower follows by
  the same logic, using concavity of the pointwise minimum. The example in Section~\ref{sec:illustration} shows that the lower (upper) bound based on $(\boldsymbol{X}, \boldsymbol{G})$ may be strictly greater (smaller).

\subsection{Proof of Corollary~\ref{cor:width}}

This follows immediately by Proposition~\ref{prop:width}, replacing $\boldsymbol{X}$ with $\emptyset$, and $\boldsymbol{G}$ with $\boldsymbol{X}$.

\subsection{Simplification under Monotonicity}

The bounds
$\gamma_{\ell}(\boldsymbol{X}) = \max_{1 \leq j \leq 8}\theta_{\ell,
  j}(\boldsymbol{X})$ and
$\gamma_{u}(\boldsymbol{X}) = \min_{1 \leq j \leq 8}\theta_{u,
  j}(\boldsymbol{X})$ on the CATE due to
\citet{balke1997bounds} are often compared to simpler valid bounds derived earlier in
\citet{robins1989, manski1990}: in our setting with baseline
confounders,
$\beta_{\ell}(\boldsymbol{X}) \leq \mathbb{E}(Y(1) - Y(0) \mid
\boldsymbol{X}) \leq \beta_u(\boldsymbol{X})$, where
\[\beta_{\ell}(\boldsymbol{X}) =
  \mathbb{E}_P(AY \mid \boldsymbol{X}, Z = 1) - \mathbb{E}_P(Y(1-A)
  + A \mid \boldsymbol{X}, Z = 0),\]
and
\[\beta_{u}(\boldsymbol{X}) = \mathbb{E}_P(AY + 1 - A\mid
  \boldsymbol{X}, Z = 1) - \mathbb{E}_P(Y(1-A)\mid \boldsymbol{X}, Z
  = 0).\] \citet{balke1997bounds} show that their bounds on the treatment
effect are in general tighter than these simpler bounds, as seen in
the following result.
\begin{lemma}
  $\beta_{\ell}(\boldsymbol{X}) = \theta_{\ell, 1}(\boldsymbol{X}) \leq
  \gamma_{\ell}(\boldsymbol{X}) \leq \gamma_{u}(\boldsymbol{X}) \leq
  \theta_{u, 1}(\boldsymbol{X}) = \beta_{u}(\boldsymbol{X})$.
\end{lemma}
\begin{proof}
  Observe that, noting $A = YA + (1 - Y)A$,
  \begin{align*}
    \beta_{\ell}(\boldsymbol{X})
    &=
      \mathbb{E}_P(AY \mid \boldsymbol{X}, Z = 1) - \mathbb{E}_P(Y(1-A)
      + A \mid \boldsymbol{X}, Z = 0) \\
    &= \pi_{11.1}(\boldsymbol{X}) - \pi_{10.0}(\boldsymbol{X}) -
      \pi_{11.0}(\boldsymbol{X}) - \pi_{01.0}(\boldsymbol{X}) \\
    &= \pi_{11.1}(\boldsymbol{X}) - \{1 - \pi_{00.0}(\boldsymbol{X})\} \\
    &= \theta_{\ell, 1}(\boldsymbol{X}).
  \end{align*}
  Similarly,
  \begin{align*}
    \beta_{u}(\boldsymbol{X})
    &=
      \mathbb{E}_P(AY + 1 - A\mid
      \boldsymbol{X}, Z = 1) -
      \mathbb{E}_P(Y(1-A)\mid \boldsymbol{X}, Z
      = 0) \\
    &= \pi_{11.1}(\boldsymbol{X}) + \pi_{10.1}(\boldsymbol{X}) +
      \pi_{00.1}(\boldsymbol{X}) - \pi_{10.0}(\boldsymbol{X})\\
    &=  \{1 - \pi_{01.1}(\boldsymbol{X})\} - \pi_{10.0}(\boldsymbol{X}) \\
    &= \theta_{u, 1}(\boldsymbol{X}).
  \end{align*}
\end{proof}
\noindent These simpler bounds are known to be tight --- see Theorem 7.3
of \citet{balke2011nonparametric} --- when we additionally assert the
monotonicity assumption, i.e., $A(z = 1) \geq A(z = 0)$ almost
surely. Thus, somewhat surprisingly, monotonicity does not
enable us to obtain tighter bounds, but does simplify the structure in
that $\beta_{\ell}(\boldsymbol{X}) = \gamma_{\ell}(\boldsymbol{X})$
and $\beta_{u}(\boldsymbol{X}) = \gamma_{u}(\boldsymbol{X})$. A consequence of this analysis is that in the
special case of a randomized trial with monotonicity,
covariate-assisted bounds offer no reduction in width for bounding
the marginal ATE compared to covariate-agnostic bounds
\[\beta_{\ell}^*(P) = \mathbb{E}_P(AY \mid Z = 1) - \mathbb{E}_P(Y(1-A)
  + A \mid Z = 0),\] and
\[\beta_{u}^*(P) =\mathbb{E}_P(AY + 1 - A\mid
  Z = 1) - \mathbb{E}_P(Y(1-A)\mid Z = 0),\] since
$\beta_{\ell}^*(P) = \mathbb{E}_P(\beta_{\ell}(\boldsymbol{X}))$ and
$\beta_{u}^*(P) = \mathbb{E}_P(\beta_{u}(\boldsymbol{X}))$ by
randomization, i.e., $Z \ind \boldsymbol{X}$. That said, one upside is that we can leverage randomization to
obtain more efficient estimators of the bounds in this setting
compared to nonparametric estimators of $\beta_{\ell}^*(P)$ and
$\beta_{u}^*(P)$. Indeed, the efficient influence functions of the bounds in this
setting will be the nonparametric influence functions of the
functionals $\mathbb{E}_P(\beta_{\ell}(\boldsymbol{X}))$ and
$\mathbb{E}_P(\beta_{u}(\boldsymbol{X}))$---we omit this analysis here.

\section{Proofs of Results in Section~\ref{sec:binary}} \label{sec:appendix-binary}
\subsection{Proof of Theorem \ref{thm:convergence_margin}}
The proof structures follows that of the proof of Theorem \ref{thm:dr-smooth}, explained in greater detail below. We prove the result for the lower bound as the result for the upper bound is analogous. By the standard decomposition, we have
\begin{align*}
\widehat{\mathcal{L}} - \mathcal{L} & = (\Pn - \Pb)\{ \varphi_\ell(O; \widehat\Pb, \widehat{d}_\ell) - \varphi_\ell(O; \Pb, d_\ell)\} + \Pb\{ \varphi_\ell(O; \widehat\Pb, \widehat{d}_\ell) - \varphi_\ell(O; \Pb, d_\ell)\} \\
& \hphantom{=} + (\Pn - \Pb)\{\varphi_\ell(O; \Pb, d_\ell)\} \\
& \equiv R_1 + R_2 + (\Pn - \Pb)\{\varphi_\ell(O; \Pb, d_\ell)\}
\end{align*}
We will show that $R_1 = o_\Pb(n^{-1/2})$ and 
\begin{align*}
R_2 = O_\Pb\left(\left\lVert \widehat{\lambda}_1 - \lambda_1 \right\rVert
    \cdot \max_{y,a,z \in \{0,1\}} \left\lVert \widehat{\pi}_{ya.z} -
      \pi_{ya.z}\right\rVert + \max_{1 \leq j \leq 8}\left\lVert
      \widehat{\theta}_{\ell, j} - \theta_{\ell, j} \right\rVert_{\infty}^{1 +
      \alpha}\right)
\end{align*}
under the conditions of the theorem.
\subsubsection{Term $R_1$}
By Lemma 2 in \cite{kennedy2020b}, $R_1=o_\Pb(n^{-1/2})$
if 
\begin{align*}
\int \{\varphi_\ell(o; \widehat\Pb, \widehat{d}_\ell) - \varphi_\ell(o; \Pb, d_\ell)\}^2 d\Pb(o) = o_\Pb(1)
\end{align*}
We have
\begin{align*}
\int \{\varphi_\ell(o; \widehat\Pb, \widehat{d}_\ell) - \varphi_\ell(o; \Pb, d_\ell)\}^2 d\Pb(o) & \lesssim \int \{\varphi_\ell(o; \widehat\Pb, \widehat{d}_\ell) - \varphi_\ell(o; \Pb, \widehat{d}_\ell)\}^2 d\Pb(o) \\
& \hphantom{=} + \int \{\varphi_\ell(o; \Pb, \widehat{d}_\ell) - \varphi_\ell(o; \Pb, d_\ell)\}^2 d\Pb(o)
\end{align*}
For the first term,
\begin{align*}
\int \{\varphi_\ell(o; \widehat\Pb, \widehat{d}_\ell) - \varphi_\ell(o; \Pb, \widehat{d}_\ell)\}^2 d\Pb(o) \lesssim \sum_{j = 1}^8 \int \{L_j(o; \widehat\Pb) - L_j(o; \Pb) + \widehat\theta_{\ell, j} - \theta_{\ell, j}\}^2 d\Pb(o) = o_\Pb(1)
\end{align*}
since for example for $j = 1$, we have
  \begin{align*}
    &\left\lVert \left(\widehat{\pi}_{11.1} - \pi_{11.1}\right)
      \left(1 - \frac{\mathds{1}(Z = 1)}{\widehat{\lambda}_1}\right)
      + \frac{\mathds{1}(Z = 1)}{\widehat{\lambda}_1\lambda_1}\left\{\mathds{1}(Y = 1, A = 1) -
      \pi_{11.1}\right\}\left(\lambda_1 - \widehat{\lambda}_1\right)\right\rVert \\
    & \quad + \left\lVert \left(\widehat{\pi}_{00.0} - \pi_{00.0}\right)
      \left(1 - \frac{\mathds{1}(Z = 0)}{\widehat{\lambda}_0}\right)
      + \frac{\mathds{1}(Z = 0)}{\widehat{\lambda}_0\lambda_0}\left\{\mathds{1}(Y = 0, A = 0) -
      \pi_{00.0}\right\}\left(\lambda_0 - \widehat{\lambda}_0\right)\right\rVert \\
    & \lesssim \left\lVert \widehat{\lambda}_1 - \lambda_1 \right\rVert +
      \max_{y,a,z \in \{0,1\}} \left\lVert \widehat{\pi}_{ya.z} -
      \pi_{ya.z}\right\rVert = o_P(1),
  \end{align*}
  by our assumptions, using the fact that $\widehat{\lambda}_z$ and
  $\lambda_z$ are bounded away from zero.

  For the second term, we have
  \begin{align*}
  \int \{\varphi_\ell(o; \Pb, \widehat{d}_\ell) - \varphi_\ell(o; \Pb, d_\ell) \}^2 d\Pb(o) & = \sum_{j = 1}^8 \int\left| \mathds{1}\{\widehat{d}_\ell(\bx) = j \} - \mathds{1}\{d_\ell(\bx) = j\} \right| \{L_j(o; \Pb) + \theta_{\ell, j}(\bx)\}^2 d\Pb(o)  \\
  & \lesssim \Pb\left\{\theta_{\ell, \widehat{d}(\boldsymbol{X})}(\boldsymbol{X}) \neq \theta_{\ell, d_\ell(\boldsymbol{X})} (\boldsymbol{X}) \right\}
  \end{align*}
  since $\theta_{\ell, j}(\boldsymbol{X})$ and $L_j(O;\Pb)$ are
  all uniformly bounded. Next, we show that
  \begin{align*}
  \Pb\left\{\theta_{\ell, \widehat{d}(\boldsymbol{X})}(\boldsymbol{X}) \neq \theta_{\ell, d_\ell(\boldsymbol{X})} (\boldsymbol{X}) \right\} = o_\Pb(1)
  \end{align*}
  For any $t > 0$, we have
\begin{align*}
    \Pb\left[\theta_{\ell, \widehat{d}_{\ell}(\boldsymbol{X})} \neq
    \theta_{\ell,d_{\ell}(\boldsymbol{X})}\right]
    &= \Pb\left[\theta_{\ell, \widehat{d}_{\ell}(\boldsymbol{X})} \neq
      \theta_{\ell,d_{\ell}(\boldsymbol{X})}, \min_{j \neq d_\ell(\boldsymbol{X})} \{\theta_{\ell,
      d_{\ell}(\boldsymbol{X})}(\boldsymbol{X}) -
      \theta_{\ell,j}(\boldsymbol{X})\} \leq t \right] \\
    & \quad \quad \quad +
      \Pb\left[\theta_{\ell, \widehat{d}_{\ell}(\boldsymbol{X})} \neq
      \theta_{\ell,d_{\ell}(\boldsymbol{X})}, \min_{j \neq d_\ell(\boldsymbol{X})} \{\theta_{\ell,
      d_{\ell}(\boldsymbol{X})}(\boldsymbol{X}) -
      \theta_{\ell,j}(\boldsymbol{X})\} > t \right] \\
    & \leq \Pb\left[\min_{j \neq d_\ell(\boldsymbol{X})} \{\theta_{\ell,
      d_{\ell}(\boldsymbol{X})}(\boldsymbol{X}) -
      \theta_{\ell,j}(\boldsymbol{X})\} \leq t \right] +
      \Pb\left[\theta_{\ell, d_{\ell}(\boldsymbol{X})} -
      \theta_{\ell,\widehat{d}_{\ell}(\boldsymbol{X})} > t\right] \\
    & \leq Ct^{\alpha} + \Pb\left[\theta_{\ell, d_{\ell}(\boldsymbol{X})} -
      \theta_{\ell,\widehat{d}_{\ell}(\boldsymbol{X})} +
      \widehat{\theta}_{\ell, \widehat{d}_{\ell}(\boldsymbol{X})} -
      \widehat{\theta}_{\ell, d_{\ell}(\boldsymbol{X})}> t\right] \\
    & \leq Ct^{\alpha} + \Pb\left\{2 \sum_{j=1}^8
      |\widehat{\theta}_{\ell, j} - \theta_{\ell, j}| > t\right\} \\
    & \leq Ct^{\alpha} + \frac{2}{t}\sum_{j=1}^8
      \Pb|\widehat{\theta}_{\ell,j}(\boldsymbol{X}) -
      \theta_{\ell, j}(\boldsymbol{X})| \\
    & \leq Ct^{\alpha} + \frac{2}{t}\sum_{j=1}^8 \left\lVert
      \widehat{\theta}_{\ell,j} - \theta_{\ell, j}\right\rVert
  \end{align*}
  where $C > 0$ is the universal constant in Assumption
  \ref{ass:margin}. In the second line, we use that
  $\theta_{\ell, \widehat{d}_{\ell}(\boldsymbol{X})} \neq
  \theta_{\ell,d_{\ell}(\boldsymbol{X})}$ implies
  $\widehat{d}_{\ell}(\boldsymbol{X}) \neq 
  d_\ell(\boldsymbol{X})$, so
  $\theta_{\ell,d_{\ell}(\boldsymbol{X})} -
  \theta_{\ell,\widehat{d}_{\ell}(\boldsymbol{X})} \geq \min_{j \neq 
   d_\ell(\boldsymbol{X})} \{\theta_{\ell,
    d_{\ell}(\boldsymbol{X})}(\boldsymbol{X}) -
  \theta_{\ell,j}(\boldsymbol{X})\}$. In the third line we use Assumption
  \ref{ass:margin} and that
  $\widehat{\theta}_{\ell,
    \widehat{d}_{\ell}(\boldsymbol{X})}(\boldsymbol{X}) -
  \widehat{\theta}_{\ell, d_{\ell}(\boldsymbol{X})}(\boldsymbol{X}) \geq 0$
  by construction of $\widehat{d}_{\ell}(\boldsymbol{X})$. The fourth
  line follows from Markov's inequality, and the last line from
  $\lVert \, \cdot \, \rVert_{L_1} \leq \lVert \, \cdot \,
  \rVert_{L_2}$. Since for each $j \in \{1,\ldots, 8\}$ we have
  $\left\lVert \widehat{\theta}_{\ell,j} - \theta_{\ell, j}\right\rVert =
  o_P(1)$, as each $\widehat{\theta}_{\ell, j} - \theta_{\ell, j}$ is a linear
  combination of the differences
  $\{\widehat{\pi}_{ya.z} - \pi_{ya.z}: y, a, z \in \{0,1\}\}$, we
  obtain the desired result by invoking Lemma
  \ref{lemma:convbound}. Namely, we set
  $X_n = \Pb\left[\theta_{\ell, \widehat{d}_{\ell}(\boldsymbol{X})} \neq
    \theta_{\ell,d_{\ell}(\boldsymbol{X})}\right]$, and for any
  $\epsilon > 0$, choose
  $t_{\epsilon} = \left(\frac{\epsilon}{C}\right)^{1/\alpha} > 0$ and
  $Z_n^{(\epsilon)} = \frac{2}{t_{\epsilon}}\sum_{j=1}^8 \left\lVert
    \widehat{\theta}_{\ell,j} - \theta_{\ell, j}\right\rVert$.
  \begin{lemma} \label{lemma:convbound}
  Suppose that for a given sequence $X_n$, one can find for any
  $\epsilon > 0$ another sequence $Z_n^{(\epsilon)} \geq 0$ such that
  $|X_n| \leq \epsilon + Z_n^{(\epsilon)}$ and
  $Z_n^{(\epsilon)} = o_P(1)$. Then $X_n = o_P(1)$.
\end{lemma}

\begin{proof}
  Fixing $\epsilon > 0$, consider a non-negative sequence
  $Z_n^{(\epsilon / 2)} = o_P(1)$ satisfying
  $|X_n| \leq \epsilon / 2 + Z_n^{(\epsilon / 2)}$. Then
  \[P[|X_n| > \epsilon] \leq P[\epsilon / 2 + Z_n^{(\epsilon / 2)} >
    \epsilon] = P[Z_n^{(\epsilon / 2)} > \epsilon / 2] \to 0 \text{ as
    } n \to \infty,\] since $Z_n = o_P(1)$, thus proving the result.
\end{proof}
\subsubsection{Term $R_2$}
We decompose $R_2$ as
  \begin{align*}
    R_2 & =  \sum_{j = 1}^8 \left\{ \Pb\left[\mathds{1}\{\widehat{d}_\ell(\boldsymbol{X}) = j\}\{L_j(O; \widehat\Pb) + \widehat\theta_{\ell, j}(\boldsymbol{X}) - \theta_{\ell, j}(\boldsymbol{X})\} \right] \right. \\
    & \left. \hphantom{=}  + \Pb\left(\left[\mathds{1}\{\widehat{d}_\ell(\boldsymbol{X}) = j\} -\mathds{1}\{d_\ell(\boldsymbol{X}) = j\} \right]\theta_{\ell, j}(\boldsymbol{X})\right) \right\} 
  \end{align*}
  We have
  \begin{align*}
  \Pb\left[\mathds{1}\{\widehat{d}_\ell(\boldsymbol{X}) = j\}\{L_j(O; \widehat\Pb) + \widehat\theta_{\ell, j}(\boldsymbol{X}) - \theta_{\ell, j}(\boldsymbol{X})\} \right] \lesssim  \left\lVert \widehat{\lambda}_1 - \lambda_1 \right\rVert
    \cdot \max_{y,a,z \in \{0,1\}} \left\lVert \widehat{\pi}_{ya.z} -
      \pi_{ya.z}\right\rVert
\end{align*}
For the second term, observe that
  \begin{align*}
    & \left|\Pb\left\{
      \theta_{\ell, \widehat{d}_{\ell}(\boldsymbol{X})}(\boldsymbol{X}) -
      \theta_{\ell, d_{\ell}(\boldsymbol{X})}(\boldsymbol{X})
      \right\} \right| \\
    & = \left|\Pb\left[\mathds{1}\{\theta_{\ell, d_{\ell}(\boldsymbol{X})}(\boldsymbol{X})
      > \theta_{\ell, \widehat{d}_{\ell}(\boldsymbol{X})}(\boldsymbol{X})\}\left\{
      \theta_{\ell, d_{\ell}(\boldsymbol{X})}(\boldsymbol{X}) -
      \theta_{\ell, \widehat{d}_{\ell}(\boldsymbol{X})}(\boldsymbol{X})\right\}
      \right] \right|\\
    & \leq \Pb\bigg(\mathds{1}\left(
      \min_{j \neq d_\ell(\boldsymbol{X})} \{\theta_{\ell,
      d_{\ell}(\boldsymbol{X})} -
      \theta_{\ell,j}\} \leq \theta_{\ell, d_{\ell}(\boldsymbol{X})}
      - \theta_{\ell, \widehat{d}_{\ell}(\boldsymbol{X})} +
      \widehat{\theta}_{\ell, \widehat{d}_{\ell}(\boldsymbol{X})} -
      \widehat{\theta}_{\ell, d_{\ell}(\boldsymbol{X})}\right) \\
    & \quad \quad \times\left(
      \theta_{\ell, d_{\ell}(\boldsymbol{X})} -
      \theta_{\ell, \widehat{d}_{\ell}(\boldsymbol{X})} +
      \widehat{\theta}_{\ell, \widehat{d}_{\ell}(\boldsymbol{X})} -
      \widehat{\theta}_{\ell, d_{\ell}(\boldsymbol{X})}\bigg)
      \right) \\
    & \leq 2 \max_{1 \leq j \leq 8}\left\lVert
      \widehat{\theta}_{\ell, j} - \theta_{\ell, j} \right\rVert_{\infty}
      \Pb\left[\min_{j \neq d_\ell(\boldsymbol{X})} \{\theta_{\ell,
      d_{\ell}(\boldsymbol{X})} -
      \theta_{\ell,j}\} \leq 2 \max_{1 \leq j \leq 8}\left\lVert
      \widehat{\theta}_{\ell, j} - \theta_{\ell, j} \right\rVert_{\infty}\right] \\
    & \lesssim \max_{1 \leq j \leq 8}\left\lVert
      \widehat{\theta}_{\ell, j} - \theta_{\ell, j} \right\rVert_{\infty}^{1 +
      \alpha},
  \end{align*}
  by Assumption \ref{ass:margin}, where we used the fact that
  $\theta_{\ell, d_{\ell}(\boldsymbol{X})}(\boldsymbol{X}) \geq \theta_{\ell,
    \widehat{d}_{\ell}(\boldsymbol{X})}(\boldsymbol{X})$ and
  $\widehat{\theta}_{\ell, d_{\ell}(\boldsymbol{X})}(\boldsymbol{X}) \leq
  \widehat{\theta}_{\ell,
    \widehat{d}_{\ell}(\boldsymbol{X})}(\boldsymbol{X})$, by
  construction of $d_{\ell}$ and $\widehat{d}_{\ell}$.

\section{Elaboration and Proofs of Results in Section~\ref{sec:smooth}} \label{sec:appendix-smooth}
\subsection{Doubly Robust Machine Learning Framework}
We now briefly review the statistical framework we use to derive our estimators and to evaluate the theoretical properties of our methods. A central goal is to develop methods that are flexible and resistant to bias from model misspecification.

Here, we reduce the potential for model misspecification by using flexible nonparametric machine learning (ML) tools. More specifically, we aim to construct bias-corrected estimators using influence functions---a central element of semiparametric theory. As an example, we review estimation of the mean counterfactual, $\E(Y(a=1))$ (i.e., the mean outcome if every unit in the population were treated), from this perspective. For the purposes of illustration, we assume (just for this paragraph) that the data consist of $n$ iid copies of $O = (\boldsymbol{X}, A, Y) \sim \mathbb{P}$. Under no unmeasured confounding and other assumptions, $\E(Y(a=1))$ can be written as an averaged regression function \citep{robins1986}: $\psi(\mathbb{P}) = \E_{\mathbb{P}}\{\E_{\mathbb{P}}(Y \mid A = 1, \boldsymbol{X})\}$. Estimating $\E_{\mathbb{P}}(Y \mid A = 1, \boldsymbol{X})$ is a standard regression problem and flexible ML methods may be preferred to more restrictive parametric estimation methods such as linear regression to avoid model misspecification and reduce bias. However, if $\psi$ is estimated as the average of predictions from a ML model, rather than from a parametric model, the estimate will generally inherit first-order smoothing bias from the nonparametric estimates. Instead, one can find a function of the data, which we can denote generically as $\varphi$, such that estimating $\psi$ as the average value of (an estimated version of) $\varphi$  will correct this first-order bias. An optimal choice of this function is referred to as the \textit{influence function} of $\psi$, and in this case is based on two regression functions, $\E_{\mathbb{P}}(Y \mid A = 1, \boldsymbol{X})$ and $\Pb(A = 1 \mid \boldsymbol{X})$. In practice, the analyst fits two models: a model for the outcome regressed against the treatment and all confounders (the outcome model); and a model regressing the treatment against all confounders (the propensity model).  These models are combined to estimate the effect of interest. This approach is ``doubly robust'', since it is consistent if either the propensity model or the outcome model is correctly specified \citep{Scharfstein:1999a}, and also leads to ``doubly reduced'' second-order bias. Finally, the estimation process is combined with sample-splitting or cross-fitting, to prevent over-fitting (or formally, to avoid imposing complexity restrictions on the class of nuisance estimators) by separating estimation of the components of the influence function from estimation of its mean \citep{robins2008higher, zheng2010asymptotic, chernozhukov2018double}. 

In semiparametric efficiency theory \citep{bickel1993efficient, tsiatis2006semiparametric, van2000asymptotic, kennedy2016semiparametric}, a fundamental goal is to characterize the (efficient) influence function. Mathematically, an influence function is the derivative in a von Mises expansion of the target statistical functional (analogous to the usual derivative of a function in Taylor expansion). In robust statistics, it coincides with the Gateaux derivative of the functional in the direction of a point-mass contamination distribution. The influence function serves a number of purposes. First, the variance of the efficient influence function is equal to the efficiency bound of the target statistical functional, which serves as a lower bound of the variance for regular estimators. It characterizes the inherent estimation difficulty of the target functional and provides a benchmark to compare against when we construct estimators. Moreover, it enables us to correct for first-order bias in the plug-in estimator and motivates the robust estimator, which has a general second-order bias property so that nonparametric and flexible machine learning methods with relatively slow rates can be used for estimating the nuisance functions.

\subsection{Background on log-sum-exp Function}
The log-sum-exp (LSE) function is commonly employed across a range of disciplines, including statistical mechanics \citep{aldous2005spin} and machine learning \citep{calafiore2020universal}. Using the LSE function to approximate the maximum function is also common in the statistics literature. For instance, one way to prove Sudakov-Fernique's inequality on Gaussian comparison is to apply the functional form of Slepian's inequality (which requires the function considered to be twice-differentiable) to the LSE function and let $t\rightarrow \infty$ \citep{vershynin2018high, wainwright2019high}. As another example, in proving a high-dimensional Gaussian comparison version of the central limit theorem, \citet{chernozhukov2012central} used Slepian's Gaussian interpolation together with Stein’s leave-one-out expansions. A typical tool to evaluate Stein’s leave-one-out expansion is a Taylor's expansion, which requires differentiability of the function. Hence, \citet{chernozhukov2012central} also approximated the maximum function with the LSE function, and selected the tuning parameter $t$ to limit approximation error while controlling the derivatives of the LSE function.

\subsection{Proof of Theorem \ref{thm:IF-smooth}}
Recall an influence function of a pathwise differentiable
functional $\chi(\Pb)$, at $\Pb$ in a given statistical model, is a zero-mean finite-variance function $\dot{\chi}(O ; \Pb)$ of observed data $O$ such that for any regular one-dimensional parametric submodel $\Pb_\epsilon$ through $\Pb_0 = \Pb$, it holds that
\[
\left.\frac{d}{d \epsilon} \chi\left(\Pb_\epsilon\right)\right|_{\epsilon=0}=\mathbb{E}_{\Pb}[\dot{\chi}(O ; \Pb) u(O)]
\]
where $u(O)$ is the score function of the parametric submodel at $\Pb$. For such a parametric submodel, invoking the total derivative, we have (let $\boldsymbol{\theta}_{\epsilon, \ell}(\boldsymbol{X})$ be the nuisance function vector at submodel $\Pb_{\epsilon}$)
\begin{equation}\label{eq:derivative-Psig}
\begin{aligned}
&\, \left.\frac{d}{d \epsilon} \mathcal{L}_g\left(\Pb_\epsilon\right)\right|_{\epsilon=0} \\
= &\, \left.\frac{d}{d \epsilon} \mathbb{E}_{\Pb_\epsilon}\left[g\left(\boldsymbol{\theta}_{\epsilon, \ell}(\boldsymbol{X})\right)\right]\right|_{\epsilon=0} \\
=&\, \left.\frac{d}{d \epsilon} \mathbb{E}_{\Pb_\epsilon}\left[g\left(\boldsymbol{\theta}_{\ell}(\boldsymbol{X})\right)\right]\right|_{\epsilon=0}+\left.\sum_{j=1}^8 \frac{d}{d \epsilon} \mathbb{E}_{\Pb}\left[g\left(\theta_{\ell, 1}(\boldsymbol{X}), \ldots, \theta_{\epsilon, \ell, j}(\boldsymbol{X}), \ldots, \theta_{\ell, 8}(\boldsymbol{X})\right)\right]\right|_{\epsilon=0} \\
=&\, \mathbb{E}_{\Pb}\left[\left(g\left(\boldsymbol{\theta}_{\ell}(\boldsymbol{X})\right)-\mathcal{L}_g(\Pb)\right) u(O)\right]+\sum_{j=1}^8 \mathbb{E}_{\Pb}\left[\left.\frac{\partial g\left(\boldsymbol{\theta}_{\ell}(\boldsymbol{X})\right)}{\partial \theta_{\ell, j}(\boldsymbol{X})} \frac{d}{d \epsilon} \theta_{\epsilon, \ell, j}(\boldsymbol{X})\right|_{\epsilon=0}\right]
\end{aligned}    
\end{equation}
using the fact that score function $u$ has mean zero to center the first term, and the chain rule for the remaining
eight terms. Next, observe that, for any $(y, a, z) \in\{0,1\}^3$, (denote in general $u(B \mid C)$ as the conditional score function for the distribution of $B$ given $C$.
\begin{equation}\label{eq:pi-derivative}
\begin{aligned}
\left.\frac{d}{d \epsilon} \pi_{\epsilon, y a . z}(\boldsymbol{X})\right|_{\epsilon=0} & =\mathbb{E}_{\Pb}\left[\left(\mathds{1}(Y=y, A=a)-\pi_{y a . z}(\boldsymbol{X})\right) u(Y, A \mid Z=z, \boldsymbol{X}) \mid Z=z, \boldsymbol{X}\right] \\
& =\mathbb{E}_{\Pb}\left[\frac{\mathds{1}(Z=z)}{\lambda_z(\boldsymbol{X})}\left\{\mathds{1}(Y=y, A=a)-\pi_{y a . z}(\boldsymbol{X})\right\} u(Y, A \mid Z=z, \boldsymbol{X}) \mid \boldsymbol{X}\right] \\
& =\mathbb{E}_{\Pb}\left[\frac{\mathds{1}(Z=z)}{\lambda_z(\boldsymbol{X})}\left\{\mathds{1}(Y=y, A=a)-\pi_{y a . z}(\boldsymbol{X})\right\} u(Y, A \mid Z, \boldsymbol{X}) \mid \boldsymbol{X}\right] \\
& =\mathbb{E}_{\Pb}\left(\psi_{y a . z}(O ; \Pb) u(Y, A \mid Z, \boldsymbol{X}) \mid \boldsymbol{X}\right),
\end{aligned}    
\end{equation}
where the first equation follows from directly taking derivatives and centering with $\pi_{ya.z}$ since $u(Y,A \mid Z=z,\boldsymbol{X})$ has conditional mean zero. The second equation follows from conditioning on $Z,\boldsymbol{X}$ and using property of conditional expectations. The last equation follows from the definition of $\psi_{ya.z}$.

Since $\theta_{\ell,j}$'s are linear combinations of $\pi_{ya.z}$ (and constant 1), we can prove
\[
\mathbb{E}_{\Pb}\left[\left.\frac{\partial g\left(\boldsymbol{\theta}_{\ell}(\boldsymbol{X})\right)}{\partial \theta_{\ell, j}(\boldsymbol{X})} \frac{d}{d \epsilon} \theta_{\epsilon, \ell, j}(\boldsymbol{X})\right|_{\epsilon=0}\right] =\mathbb{E}_{\Pb}\left[\frac{\partial g\left(\boldsymbol{\theta}_{\ell}(\boldsymbol{X})\right)}{\partial \theta_{\ell, j}(\boldsymbol{X})} L_j(O ; \Pb) u(O)\right].
\]
For illustration, take $j=1$ (other $j$'s can be proved similarly) and we have
\[
\begin{aligned}
&\, \mathbb{E}_{\Pb}\left[\left.\frac{\partial g\left(\boldsymbol{\theta}_{\ell}(\boldsymbol{X})\right)}{\partial \theta_{\ell, 1}(\boldsymbol{X})} \frac{d}{d \epsilon} \theta_{\epsilon, \ell, 1}(\boldsymbol{X})\right|_{\epsilon=0}\right] \\
= &\,\mathbb{E}_{\Pb}\left[\left.\frac{\partial g\left(\boldsymbol{\theta}_{\ell}(\boldsymbol{X})\right)}{\partial \theta_{\ell, 1}(\boldsymbol{X})} \frac{d}{d \epsilon}\left\{\pi_{\epsilon, 11.1}(\boldsymbol{X})+\pi_{\epsilon, 00.0}(\boldsymbol{X})-1\right\}\right|_{\epsilon=0}\right] \\
=&\,\mathbb{E}_{\Pb}\left[\frac{\partial g\left(\boldsymbol{\theta}_{\ell}(\boldsymbol{X})\right)}{\partial \theta_{\ell, 1}(\boldsymbol{X})} \mathbb{E}_{\Pb}\left[\left\{\psi_{11.1}(O ; \Pb)+\psi_{00.0}(O ; \Pb)\right\} u(Y, A \mid Z, \boldsymbol{X}) \mid \boldsymbol{X}\right]\right] \\
=& \,\mathbb{E}_{\Pb}\left[\frac{\partial g\left(\boldsymbol{\theta}_{\ell}(\boldsymbol{X})\right)}{\partial \theta_{\ell, 1}(\boldsymbol{X})}\left\{\psi_{11.1}(O ; \Pb)+\psi_{00.0}(O ; \Pb)\right\} u(Y, A \mid Z, \boldsymbol{X})\right] \\
=& \,\mathbb{E}_{\Pb}\left[\frac{\partial g\left(\boldsymbol{\theta}_{\ell}(\boldsymbol{X})\right)}{\partial \theta_{\ell, 1}(\boldsymbol{X})} L_1(O ; \Pb) u(O)\right],
\end{aligned}
\]
where the first equation follows from definition of $\theta_{\epsilon, \ell, 1}$. The second equation follows from \eqref{eq:pi-derivative}. The third equation follows from property of conditional distribution. In the final equation we added $u(Z,\boldsymbol{X})$ and noted $u(O) =
u(Z,\boldsymbol{X}) + u(Y,A \mid Z,\boldsymbol{X})$ — we are permitted to add this term as it is a function only of $(Z,X)$, and $\psi_{ya.z}(O; \Pb)$ has
mean zero given $(Z,\boldsymbol{X})$.

Plugging these equations into \eqref{eq:derivative-Psig} we have
\[
\left.\frac{d}{d \epsilon} \mathcal{L}_g\left(\Pb_\epsilon\right)\right|_{\epsilon=0}=\mathbb{E}_{\Pb}\left\{\left[g\left(\boldsymbol{\theta}_{\ell}(\boldsymbol{X})\right)-\mathcal{L}_g(\Pb)+\sum_{j=1}^8 \frac{\partial g\left(\boldsymbol{\theta}_{\ell}(\boldsymbol{X})\right)}{\partial \theta_{\ell, j}(\boldsymbol{X})} L_j(O ; \Pb)\right] u(O)\right\}.
\]
The argument for deriving the influence function of $\mathcal{U}_h(\Pb)$ uses the exact same logic.

\subsection{Proof of Theorem \ref{thm:dr-smooth}}
We first prove a proposition that characterizes the conditional bias of the robust estimator. In the following derivations all the expectations are taken conditioning on training data $D^n$. 

\begin{proposition}\label{prop:dr-smooth-bias}
Suppose $g$ is a twice continuously differentiable function, then we have
\[
\begin{array}{r}
\E[\widehat{\mathcal{L}}_g] - \mathcal{L}_g=\mathbb{E}_{\Pb}\left[\nabla g\left(\widehat{\boldsymbol{\theta}}_{\ell}(\boldsymbol{X})\right)^T\left(\mathbb{E}_{\Pb}(\boldsymbol{L}(O ; \widehat{\Pb}) \mid \boldsymbol{X})+\widehat{\boldsymbol{\theta}}_{\ell}(\boldsymbol{X})-\boldsymbol{\theta}_{\ell}(\boldsymbol{X})\right)\right] \\
\quad-\frac{1}{2} \mathbb{E}_{\Pb}\left[\left(\widehat{\boldsymbol{\theta}}_{\ell}(\boldsymbol{X})-\boldsymbol{\theta}_{\ell}(\boldsymbol{X})\right)^T \nabla^2 g\left(\boldsymbol{\theta}_{\ell}^*(\boldsymbol{X})\right)\left(\widehat{\boldsymbol{\theta}}_{\ell}(\boldsymbol{X})-\boldsymbol{\theta}_{\ell}(\boldsymbol{X})\right)\right],
\end{array}
\]
where $\boldsymbol{L}(O ; \widehat{\Pb})=\left(L_1(O ; \widehat{\Pb}), \ldots, L_8(O ; \widehat{\Pb})\right)^T$ and $\boldsymbol{\theta}_{\ell}^*(\boldsymbol{X})$ is a point that lies on the line segment between $\boldsymbol{\theta}_{\ell}(\boldsymbol{X})$ and $\widehat{\boldsymbol{\theta}}_{\ell}(\boldsymbol{X})$.
\end{proposition}
\begin{proof}
By definition of $\widehat{\mathcal{L}}_g$, 
\[
\begin{aligned}
& \, \E[\widehat{\mathcal{L}}_g] - \mathcal{L}_g \\
=&\, \mathbb{E}_{\Pb}\left[g\left(\widehat{\boldsymbol{\theta}}_{\ell}(\boldsymbol{X})\right)+\sum_{j=1}^8 \frac{\partial g\left(\widehat{\boldsymbol{\theta}}_{\ell}(\boldsymbol{X})\right)}{\partial \widehat{\theta}_{\ell, j}(\boldsymbol{X})} L_j(O ; \widehat{\Pb})-g\left(\boldsymbol{\theta}_{\ell}(\boldsymbol{X})\right)\right] \\
=&\,\mathbb{E}_{\Pb}\left[g\left(\widehat{\boldsymbol{\theta}}_{\ell}(\boldsymbol{X})\right)+\sum_{j=1}^8 \frac{\partial g\left(\widehat{\boldsymbol{\theta}}_{\ell}(\boldsymbol{X})\right)}{\partial \widehat{\theta}_{\ell, j}(\boldsymbol{X})} \mathbb{E}_{\Pb}\left(L_j(O ; \widehat{\Pb}) \mid \boldsymbol{X}\right)-g\left(\boldsymbol{\theta}_{\ell}(\boldsymbol{X})\right)\right] \\
=&\,\mathbb{E}_{\Pb}\left[\nabla g\left(\widehat{\boldsymbol{\theta}}_{\ell}(\boldsymbol{X})\right)^T \mathbb{E}_{\Pb}(\boldsymbol{L}(O ; \widehat{\Pb}) \mid \boldsymbol{X})- \left(g\left(\boldsymbol{\theta}_{\ell}(\boldsymbol{X})\right) - g(\widehat{\boldsymbol{\theta}}_{\ell}(\boldsymbol{X}))\right)\right] \\
=&\, \mathbb{E}_{\Pb}\left[\nabla g\left(\widehat{\boldsymbol{\theta}}_{\ell}(\boldsymbol{X})\right)^T \left( \mathbb{E}_{\Pb}[\boldsymbol{L}(O ; \widehat{\Pb}) \mid \boldsymbol{X}] + \widehat{\boldsymbol{\theta}}_\ell(\boldsymbol{X}) -{\boldsymbol{\theta}}_\ell(\boldsymbol{X}) \right)\right] \\
&\, -\frac{1}{2} \mathbb{E}_{\Pb}\left[\left(\widehat{\boldsymbol{\theta}}_{\ell}(\boldsymbol{X})-\boldsymbol{\theta}_{\ell}(\boldsymbol{X})\right)^T \nabla^2 g\left(\boldsymbol{\theta}_{\ell}^*(\boldsymbol{X})\right)\left(\widehat{\boldsymbol{\theta}}_{\ell}(\boldsymbol{X})-\boldsymbol{\theta}_{\ell}(\boldsymbol{X})\right)\right].
\end{aligned}
\]
The second equality follows from conditioning on $\boldsymbol{X}$ and the last equality follows from second-order Taylor expansion.
\end{proof}
For any $(y, a, z) \in\{0,1\}^3$, by conditioning on $(Z, \boldsymbol{X})$ we have
\[
\begin{aligned}
\mathbb{E}_{\Pb}\left(\psi_{ya. z}(O ; \widehat{\Pb}) \mid \boldsymbol{X}\right) & =\E \left[ \frac{\mathds{1}(Z=z)}{\widehat{\lambda}_z(\boldsymbol{X})}\left(\Pb(Y=y, A=a\mid \boldsymbol{X}, Z=z )-\widehat{\pi}_{ya.z}(\boldsymbol{X}) \right) \mid \boldsymbol{X}\right] \\
& =\frac{\lambda_z(\boldsymbol{X})}{\widehat{\lambda}_z(\boldsymbol{X})}\left\{\pi_{y a . z}(\boldsymbol{X})-\widehat{\pi}_{y a . z}(\boldsymbol{X})\right\}.
\end{aligned}
\]
We want to bound each component of $\mathbb{E}_{\Pb}[\boldsymbol{L}(O ; \widehat{\Pb}) \mid \boldsymbol{X}] + \widehat{\boldsymbol{\theta}}_\ell(\boldsymbol{X}) -{\boldsymbol{\theta}}_\ell(\boldsymbol{X})$. We only analyze the first component as an illustration. All other components can be similarly analyzed.
\[
\begin{aligned}
&\, \mathbb{E}_{\Pb}\left(L_1(O ; \widehat{\Pb}) \mid \boldsymbol{X}\right)+\widehat{\theta}_{\ell, 1}(\boldsymbol{X})-\theta_{\ell, 1}(\boldsymbol{X}) \\
=&\, \mathbb{E}_{\Pb}\left(\psi_{11.1}(O ; \widehat{\Pb})+\psi_{00.0}(O ; \widehat{\Pb}) \mid \boldsymbol{X}\right)+\left\{\widehat{\pi}_{11.1}(\boldsymbol{X})-\pi_{11.1}(\boldsymbol{X})\right\}+\left\{\widehat{\pi}_{00.0}(\boldsymbol{X})-\pi_{00.0}(\boldsymbol{X})\right\} \\
=& \, \left\{1-\frac{\lambda_1(\boldsymbol{X})}{\widehat{\lambda}_1(\boldsymbol{X})}\right\}\left\{\widehat{\pi}_{11.1}(\boldsymbol{X})-\pi_{11.1}(\boldsymbol{X})\right\}+\left\{1-\frac{\lambda_0(\boldsymbol{X})}{\widehat{\lambda}_0(\boldsymbol{X})}\right\}\left\{\widehat{\pi}_{00.0}(\boldsymbol{X})-\pi_{00.0}(\boldsymbol{X})\right\} .
\end{aligned}
\]
Note that 
\[
1-\frac{\lambda_1(\boldsymbol{X})}{\widehat{\lambda}_1(\boldsymbol{X})}=\frac{\widehat{\lambda}_1(\boldsymbol{X})-\lambda_1(\boldsymbol{X})}{\widehat{\lambda}_1(\boldsymbol{X})}, 1-\frac{\lambda_0(\boldsymbol{X})}{\widehat{\lambda}_0(\boldsymbol{X})}=\frac{\lambda_1(\boldsymbol{X})-\widehat{\lambda}_1(\boldsymbol{X})}{1-\widehat{\lambda}_1(\boldsymbol{X})}
\]
By positivity assumption we have
\[
\begin{aligned}
&  \left|\mathbb{E}_{\Pb}\left(L_1(O ; \widehat{\Pb}) \mid \boldsymbol{X}\right)+\widehat{\theta}_{\ell, 1}(\boldsymbol{X})-\theta_{\ell, 1}(\boldsymbol{X}) \right|  \\
& \leq \frac{1}{\epsilon}\left|\widehat{\lambda}_1(\boldsymbol{X})-\lambda_1(\boldsymbol{X})\right| \cdot\left\{\left|\widehat{\pi}_{11.1}(\boldsymbol{X})-\pi_{11.1}(\boldsymbol{X})\right|+\left|\widehat{\pi}_{00.0}(\boldsymbol{X})-\pi_{00.0}(\boldsymbol{X})\right|\right\}
\end{aligned}
\]
Similar inequalities can be obtained for $2\leq j \leq 8$. Hence, by Hölder's inequality,
\[
\begin{aligned}
&\, \left|\mathbb{E}_{\Pb}\left(\nabla g\left(\widehat{\boldsymbol{\theta}}_{\ell}(\boldsymbol{X})\right)^T\left\{\mathbb{E}_\Pb(\boldsymbol{L}(O ; \widehat{\Pb}) \mid \boldsymbol{X})+\widehat{\boldsymbol{\theta}}_{\ell}(\boldsymbol{X})-\boldsymbol{\theta}_{\ell}(\boldsymbol{X})\right\}\right)\right| \\
\leq & \, \mathbb{E}_\Pb\left[\left\|\nabla g\left(\widehat{\boldsymbol{\theta}}_{\ell}(\boldsymbol{X})\right)\right\|_{\infty}\left\|\mathbb{E}_\Pb(\boldsymbol{L}(O ; \widehat{\Pb}) \mid \boldsymbol{X})+\widehat{\boldsymbol{\theta}}_{\ell}(\boldsymbol{X})-\boldsymbol{\theta}_{\ell}(\boldsymbol{X})\right\|_1\right] \\
\leq & \, \frac{C_1}{\epsilon} \mathbb{E}_\Pb\left[\left|\widehat{\lambda}_1(\boldsymbol{X})-\lambda_1(\boldsymbol{X})\right| \sum_{(y, a, z) \in \mathcal{R}}\left|\widehat{\pi}_{y a . z}(\boldsymbol{X})-\pi_{y a . z}(\boldsymbol{X})\right|\right],
\end{aligned}
\]
where $\mathcal{R}$ is a multiset of elements in $\{0,1\}^3$ such that each $(y, a, z) \in\{0,1\}^3$ appears in $\mathcal{R}$ as many times as $\pi_{y a . z}(\boldsymbol{X})$ appears in $\boldsymbol{\theta}_\ell(\boldsymbol{X})$. By the triangle inequality and Cauchy-Schwarz's inequality
\[
\begin{aligned}
&\, \left|\mathbb{E}_{\Pb}\left(\nabla g\left(\widehat{\boldsymbol{\theta}}_{\ell}(\boldsymbol{X})\right)^T\left\{\mathbb{E}_\Pb(\boldsymbol{L}(O ; \widehat{\Pb}) \mid \boldsymbol{X})+\widehat{\boldsymbol{\theta}}_{\ell}(\boldsymbol{X})-\boldsymbol{\theta}_{\ell}(\boldsymbol{X})\right\}\right)\right|    \\
\lesssim &\, C_1 \left\|\widehat{\lambda}_1-\lambda_1\right\| \sum_{(y, a, z) \in \mathcal{R}}\left\|\widehat{\pi}_{y a . z}-\pi_{y a . z}\right\| \\
\lesssim &\, C_1\left\|\widehat{\lambda}_1-\lambda_1\right\|\left(\max _{y, a, z \in\{0,1\}}\left\|\widehat{\pi}_{y a . z}-\pi_{y a . z}\right\|\right)
\end{aligned}
\]
It remains to bound the second derivative term in the asymptotic bias expression derived in Proposition \ref{prop:dr-smooth-bias}. By properties of the operator norm we have
\[
\begin{aligned}
& \,\left|\mathbb{E}_\Pb\left(\left(\widehat{\boldsymbol{\theta}}_{\ell}(\boldsymbol{X})-\boldsymbol{\theta}_{\ell}(\boldsymbol{X})\right)^T \nabla^2 g\left(\boldsymbol{\theta}_{\ell}^*(\boldsymbol{X})\right)\left(\widehat{\boldsymbol{\theta}}_{\ell}(\boldsymbol{X})-\boldsymbol{\theta}_{\ell}(\boldsymbol{X})\right)\right)\right| \\
\leq &\, \left\{\sup _{\boldsymbol{\theta}}\| \nabla^2 g(\boldsymbol{\theta}) \|\right\} \mathbb{E}_\Pb\left(\left\|\widehat{\boldsymbol{\theta}}_{\ell}(\boldsymbol{X})-\boldsymbol{\theta}_{\ell}(\boldsymbol{X})\right\|_2^2\right)  \\
\lesssim &\, C_2 \max _{y, a, z \in\{0,1\}} \left\|\widehat{\pi}_{ {ya.z }}-\pi_{ {ya.z }}\right\|^2
\end{aligned}
\]
The bound on conditional bias is established. For a general function $f$ on the sample $O$, we have
\begin{equation}\label{eq:decompose-error}
    \Pb_n[\widehat{f}] - \E[f] = (\Pb_n - \E) (\widehat{f} - f) + (\Pb_n -\E)(f) + \E(\widehat{f}-f)
\end{equation}
We apply the decomposition of error \eqref{eq:decompose-error} to
\[
f(O) =  \dot{\mathcal{L}}_g(O ; \Pb)  + \mathcal{L}_g( \Pb) =g\left(\boldsymbol{\theta}_{\ell}(\boldsymbol{X})\right)+\sum_{j=1}^8 \frac{\partial g\left(\boldsymbol{\theta}_{\ell}(\boldsymbol{X})\right)}{\partial \theta_{\ell, j}(\boldsymbol{X})} L_j(O ; \Pb),
\]
Note that $\Pb_n[\widehat{f}]$ is exactly the robust estimator. 
By Lemma 2 in \cite{kennedy2020b}, if $ \|\widehat{f}-f\|_2 = o_{\Pb}(1)$ we have
\[
(\Pb_n - \E) (\widehat{f} - f) = o_{\Pb}(n^{-1/2}).
\]
Also note $\E(\widehat{f}-f)$ is equal to the conditional bias and under the convergence rate assumption we have $\E(\widehat{f}-f) = o_{\Pb}(n^{-1/2})$. Finally note that $f-\E f = \dot{\mathcal{L}}_g(O ; \Pb)$, the proof is completed.

\section{Proofs of Results in Section~\ref{sec:continuous}}
\label{sec:appendix-continuous}

\subsection{Proof of Theorem~\ref{thm:continuous-cdf}}

First by Fubini's theorem and Jensen's inequality we have
\[
\int_0^1\mathbb{E} \left[\min _{1 \leq j \leq 8} {\theta}_{u, j}(t,\boldsymbol{X}) \right] d t = \mathbb{E} \left[ \int_0^1 \min _{1 \leq j \leq 8} {\theta}_{u, j}(t,\boldsymbol{X}) dt\right] \leq \mathbb{E} \left[  \min _{1 \leq j \leq 8} 
 \int_0^1{\theta}_{u, j}(t,\boldsymbol{X}) dt\right].
\]
The proof will be completed if we can show 
\[
 \int_0^1{\theta}_{u, j}(t,\boldsymbol{X}) dt = \widetilde{{\theta}}_{u, j}(\boldsymbol{X}).
\]
Since $\theta_{u,j}$'s are linear combinations of $\pi_{1a.z}$ and $\pi_{0a.z}$, we only need to show for any $a,z \in \{0,1\}$,
\[
\int_0^1\pi_{1a.z}(t,\boldsymbol{X}) dt = \widetilde{{\pi}}_{1a.z}(\boldsymbol{X}),
\]
\[
\int_0^1\pi_{0a.z}(t,\boldsymbol{X}) dt = \widetilde{{\pi}}_{0a.z}(\boldsymbol{X}).
\]
We condition on $W$, by property of conditional expectation,
\begin{equation*}
    \begin{aligned}
        &\, \widetilde{{\pi}}_{1a.z}(\boldsymbol{X})\\
        = &\, \Pb(Y \leq W, A=a \mid Z=z,\boldsymbol{X}) \\
        = &\, \int_0^1 \Pb(Y \leq t, A=a \mid Z=z,\boldsymbol{X},W=t) p_w(t \mid Z=z,\boldsymbol{X}) dt,
    \end{aligned}
\end{equation*}
where $p_w$ denotes the density of $W$. Since $W$ is independent of the data generating process of $O=(\boldsymbol{X},Z,A,Y)$, we have
\[
\Pb(Y \leq t, A=a \mid Z=z,\boldsymbol{X},W=t) = \Pb(Y \leq t, A=a \mid Z=z,\boldsymbol{X}),
\]
\[
p_w(t \mid Z=z,\boldsymbol{X}) = p_w(t) = \mathds{1}(0<t<1).
\]
Hence we conclude
\[
\int_0^1 \Pb(Y \leq t, A=a \mid Z=z,\boldsymbol{X},W=t) p_w(t \mid Z=z,\boldsymbol{X}) dt = \int_0^1 \Pb(Y \leq t, A=a \mid Z =z,\boldsymbol{X}) dt.
\]
Note that by definition $\pi_{1a.z}(t,\boldsymbol{X}) = \Pb(Y \leq t, A=a\mid\boldsymbol{X}, Z=z)$, which implies
\[
\int_0^1\pi_{1a.z}(t,\boldsymbol{X}) dt = \widetilde{{\pi}}_{1a.z}(\boldsymbol{X}).
\]
We can similarly prove
\[
\int_0^1\pi_{0a.z}(t,\boldsymbol{X}) dt = \widetilde{{\pi}}_{0a.z}(\boldsymbol{X}).
\]
which completes the proof for the first inequality in Theorem \ref{thm:continuous-cdf}. The second inequality can be proved by the same arguments. 

\subsection{Basic Efficiency Result}
\begin{proposition}\label{prop:simp}
  Suppose we observe $n$ iid copies of $T \sim P$, with $T \in
  [0,1]$. Letting $\mu = \mathbb{E}_P(T)$,
  $\mathrm{Var}_P(T) = \sigma^2$, construct $m$ estimates of $\mu$ by
  sampling $m \times n$ independent $\mathrm{Unif}(0,1)$ variates,
  $\{W_i^{(j)}\}_{i=1,\ldots,n}^{j=1, \ldots, m}$, and estimating
  $\widehat{\mu}_j = \mathbb{P}_n[T > W^{(j)}] = \frac{1}{n}
  \sum_{i=1}^n \mathds{1}(T_i > W_i^{(j)})$. Then the estimator that
  averages these $m$ estimates,
  \[\widehat{\mu} = \frac{1}{m} \sum_{j=1}^m \widehat{\mu}_j,\]
  is unbiased and has variance
  $\frac{1}{n}\left(\sigma^2 + \frac{1}{m}\left\{\mu(1-\mu) - \sigma^2
    \right\}\right) \overset{m \to \infty}{\to} \frac{\sigma^2}{n}$.
\end{proposition}

\begin{proof}
  Observe that
  $\mathbb{E}_P(\widehat{\mu}_j) = P[T > W] = \int_0^1 P[T > w] \, dw
  = \mathbb{E}_P(T) = \mu$, for $j = 1, \ldots, m$, so $\widehat{\mu}$
  is unbiased. Further,
  $\mathrm{Var}_P(\widehat{\mu}_j) =
  \frac{1}{n}\mathrm{Var}_P(\mathds{1}(T > W)) = \frac{1}{n}
  \mu(1-\mu)$, so by identical distribution of each $\widehat{\mu}_j$,
  \begin{align*}
    \mathrm{Var}_P(\widehat{\mu})
    &= \frac{1}{m} \mathrm{Var}_P(\widehat{\mu}_1)
      + \left(1 - \frac{1}{m}\right)
      \mathrm{Cov}_P(\widehat{\mu}_1, \widehat{\mu}_2) \\
    &= \frac{1}{nm} \mu(1-\mu) + \left(1 - \frac{1}{m}\right)
      \frac{1}{n}\mathrm{Cov}_P(\mathds{1}(T > W^{(1)}),
      \mathds{1}(T > W^{(2)}))
  \end{align*}
  using the fact that
  $\{W_{i}^{(j)}\} \ind (T_1, \ldots, T_n)$. Next, see that
  \[
    \mathrm{Cov}_P(\mathds{1}(T > W^{(1)}), \mathds{1}(T > W^{(2)})) =
    P[T > \max{\{W^{(1)}, W^{(1)}\}}] - \mu^2,
  \]
  and finally, since $V = \max{\{W^{(1)}, W^{(1)}\}}$ has density
  $2v\mathds{1}(v \in (0,1))$,
  \[P[T > V] = \int_0^1 2v P[T > v] \, dv = \int_0^1 P[T > \sqrt{u}]
    \, du = \int_0^1 P[T^2 > u] \, du = \mathbb{E}_P(T^2) = \mu^2 +
    \sigma^2,\] where we made the substitution $u = v^2$.
  
\end{proof}

\section{Details of the Simulation Study}
\label{sec:appendix-simulation}

Consider the following simplified setting: with a single covariate $\boldsymbol{X} \equiv X \sim \mathrm{Uniform}(0,1)$, we generated $Z \sim \mathrm{Bernoulli}(\lambda_1(X))$ based on $\lambda_1(X) = \min{\{\max{\{X^2, 0.10\}}, 0.90\}}$. Next, we assumed perfect compliance, i.e., $A \equiv Z$, then generated $U \sim \mathrm{Uniform}(0,1)$, independent of $(X, A)$, and $Y \mid X, A, U \sim \mathrm{Bernoulli}\left(U \left\{A(1 - X) + (1 - A)X\right\}\right)$.
This simplified setup, in which $X$ is sufficient for controlling for $A-Y$ confounding, results in 
\begin{align*}
\min_{j \neq d_{\ell}(X)}\left\{\theta_{\ell, d_{\ell}(X)}(X) - \theta_{\ell, j}(X)\right\} 
&= \min{\left\{\mu_0(X), 1 - \mu_0(X), \mu_1(X), 1 - \mu_1(X)\right\}} \\
&= \frac{1}{2}\min{\left\{X, 1 - X\right\}}, 
\end{align*}
where $\mu_a(X) = P[Y = 1 \mid X, A = a]$.
In turn, this guarantees that Assumption~\ref{ass:margin} holds with $\alpha = 1$, since $\left\{X, 1 - X\right\} \sim \mathrm{Uniform}\left(0,\frac{1}{2}\right)$: 
\[P\left[\frac{1}{2}\min{\left\{X, 1 - X\right\}} \leq t \right] = \begin{cases}
0 & \text{ if } t < 0 \\
4 t & \text{ if } 0 \leq t \leq \frac{1}{4} \\
1 & \text{ if } t > \frac{1}{4}
\end{cases}
\] which is at most $4t$ for all $t \geq 0$.  
Moreover, the conditional bounds satisfy $\gamma_{\ell}(X) = \gamma_{u}(X) = \frac{1}{2} - X$, and thus the marginal bounds are $\mathcal{L}(\mathbb{P}) = \mathcal{U}(\mathbb{P}) = 0$. In this scenario, the probabilities $\left\{\pi_{ya.z}(X):y,a,z \in \{0,1\}\right\}$ are given by
\[\pi_{00.1}(X) = \pi_{10.1}(X) = \pi_{01.0}(X) = \pi_{11.0}(X) = 0,\]
\[\pi_{00.0}(X) = 1 - \frac{1}{2}X, \pi_{10.0}(X) = \frac{1}{2}X, \pi_{01.1}(X) = 1 - \frac{1}{2}(1 - X), \pi_{11.1}(X) = \frac{1}{2}(1 - X).\]
To ``estimate'' the nuisance functions, we define $\widehat{\lambda}_{1}(X) = \mathrm{expit}\left(\mathrm{logit}\left(\lambda_{1}(X)\right) + \epsilon_{\lambda}\right)$, $\widehat{\pi}_{ya.z} \equiv 0$ when $a \neq z$, and otherwise $\widehat{\pi}_{ya.z}(X) = \mathrm{expit}\left(\mathrm{logit}\left(\pi_{ya.z}(X)\right) + \epsilon_{ya.z}\right)$, where \[\epsilon_\lambda, \epsilon_{00.0}, \epsilon_{10.0}, \epsilon_{01.1}, \epsilon_{11.1} \overset{\mathrm{iid}}{\sim} \mathcal{N}(h \, n^{-r}, h^2\,  n^{-2r}),\]
where we set $h = 2.25$, and vary the parameter $r \in (0,0.5]$ in different scenarios. These choices guarantee that $\left\lVert \widehat{\lambda}_{1} - \lambda_1\right\rVert = O(n^{-r})$ and $\left\lVert \widehat{\pi}_{ya.z} - \pi_{ya.z}\right\rVert = O(n^{-r})$ for $a = z$---as a consequence we can study the performance of the proposed estimators under different nuisance estimation convergence rates known to be $O(n^{-r})$.

Our simulation study proceeded as follows: for $n \in \{500, 1000, 5000\}$, we generated data and computed nuisance estimates $(\widehat{\lambda}_1, \{\widehat{\pi}_{ya.z}: y,a,z \in \{0,1\}\})$ as described above, varying $r \in \{0.10 + 0.05k: k \in \{0, \ldots, 8\}\}$. We then computed the direct lower bound estimator $\widehat{\mathcal{L}}$ as in Section~\ref{sec:binary}, and the log-sum-exp smooth approximation-based estimator $\widehat{\mathcal{L}}_{g_t}$ as described in Section~\ref{sec:smooth}, taking $t = 2 h n^r$. Finally, we computed a plug-in estimator $\mathbb{P}_n\left[\max_{1\leq j \leq 8} \widehat{\theta}_{\ell, j}(X)\right]$, whose error is expected to be first order, i.e., on the order of the nuisance error $O(n^{-r})$. Each scenario was replicated $5,000$ times, and root-mean-square error (RMSE) of each estimator relative to the lower true bound $\mathcal{L}(\mathbb{P}) = 0$ was computed. Results are shown in Figure~\ref{fig:sim-results}.

\begin{figure}[ht]
  \centering
  \includegraphics[width = \linewidth]{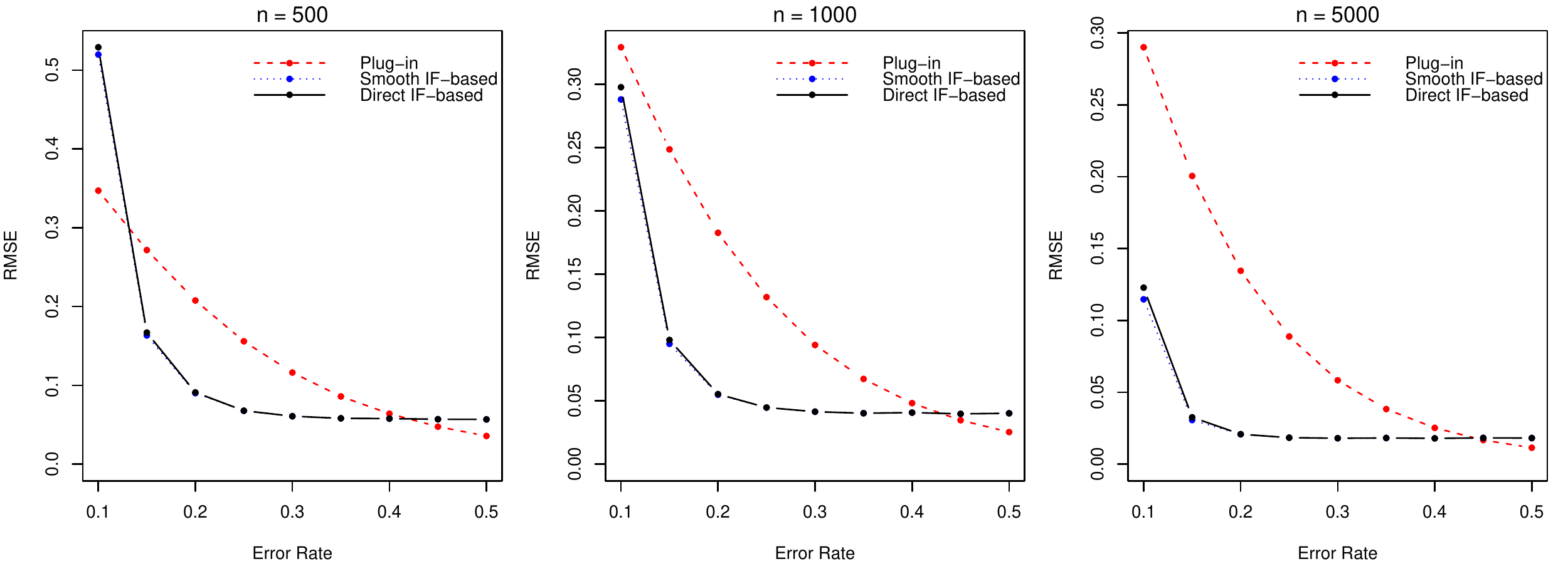}
  \caption{Root-mean-square error versus nuisance function error rate parameter $r$} \label{fig:sim-results}
\end{figure}

According to the results in Figure~\ref{fig:sim-results}, the proposed estimators $\widehat{\mathcal{L}}$ and $\widehat{\mathcal{L}}_{g_t}$ tend to outperform the plug-in estimator when nuisance functions are estimated slower than the parametric rate of $O(n^{-1/2})$, especially for larger sample sizes. Indeed, as predicted by our theoretical results, the robust estimators $\widehat{\mathcal{L}}$ and $\widehat{\mathcal{L}}_{g_t}$ nearly attained the optimal performance of the plug-in estimator (where nuisance error $\asymp n^{-1/2})$ when nuisance RMSE was only on the order $O(n^{-1/4})$, following which performance plateaued. As nuisance error on the order of $O(n^{-1/2})$ would be anticipated only in the unlikely scenario of a correctly specified parametric model, the results support our recommendation to use the proposed estimators in practice. Note that in this simulation, the \textit{ad hoc} choice for the tuning parameter $t$ in $\widehat{\mathcal{L}}_{g_t}$ resulted in very little difference compared to $\widehat{\mathcal{L}}$, though other choices might be considered.

In this simulation study, we considered only the case $\alpha = 1$ from Assumption~\ref{ass:margin}. In future work, we will consider simulation settings in which we can vary this margin parameter, and pursue theoretical development of optimal choices (e.g., in terms of minimizing MSE) for the tuning parameter $t$ in $\widehat{\mathcal{L}}_{g_t}$.

\section{Code for Reproducing Motivating Illustration}

\begin{verbatim}
library(rpart)

X1.p <- 0.7

beta.AT.0 <- function(x1, x2) {
  plogis(qlogis(0.20))
}
beta.AT.d <- function(x1, x2) {
  plogis(qlogis(0.35))
}
beta.NT.0 <- function(x1, x2) {
  plogis(qlogis(0.90))
}
beta.NT.d <- function(x1, x2) {
  plogis(qlogis(0.95))
}
beta.DE.0 <- function(x1, x2) {
  plogis(qlogis(0.65))
}
beta.DE.d <- function(x1, x2) {
  plogis(qlogis(0.725))
}
beta.CO.0 <- function(x1, x2) {
  plogis(qlogis(0.25))
}
beta.CO.d <- function(x1, x2) {
  plogis(qlogis(0.375))
}


## we work with an 8-vector pi = {pi_{ya.z} : y, a, z \in {0,1}}
## specifically, take the following order:
## pi = (pi_{00.0}, pi_{01.0}, pi_{10.0}, pi_{11.0}, 
##       pi_{00.1}, pi_{01.1}, pi_{10.1}, pi_{11.1})

## lower bound functions
p.l1 <- function(pi) { pi[8] + pi[1] - 1 }
p.l2 <- function(pi) { pi[4] + pi[5] - 1 }
p.l3 <- function(pi) { -pi[6] - pi[7] }
p.l4 <- function(pi) { -pi[2] - pi[3] }
p.l5 <- function(pi) { pi[4] - pi[8] - pi[7] - pi[2] - pi[3] }
p.l6 <- function(pi) { pi[8] - pi[4] - pi[3] - pi[6] - pi[7] }
p.l7 <- function(pi) { pi[5] - pi[6] - pi[7] - pi[2] - pi[1] }
p.l8 <- function(pi) { pi[1] - pi[2] - pi[3] - pi[6] - pi[5] }
gamma.l <- function(pi) { pmax(p.l1(pi), p.l2(pi), p.l3(pi), p.l4(pi),
                               p.l5(pi), p.l6(pi), p.l7(pi), p.l8(pi)) }

## upper bound functions
p.u1 <- function(pi) { 1 - pi[6] - pi[3] }
p.u2 <- function(pi) { 1 - pi[2] - pi[7] }
p.u3 <- function(pi) { pi[8] + pi[5] }
p.u4 <- function(pi) { pi[4] + pi[1] }
p.u5 <- function(pi) { -pi[2] + pi[6] + pi[5] + pi[4] + pi[1] }
p.u6 <- function(pi) { -pi[6] + pi[2] + pi[1] + pi[8] + pi[5] }
p.u7 <- function(pi) { -pi[7] + pi[8] + pi[5] + pi[4] + pi[3] }
p.u8 <- function(pi) { -pi[3] + pi[4] + pi[1] + pi[8] + pi[7] }
gamma.u <- function(pi) { pmin(p.u1(pi), p.u2(pi), p.u3(pi), p.u4(pi),
                               p.u5(pi), p.u6(pi), p.u7(pi), p.u8(pi)) }

set.seed(476)
N <- 1000000
X1 <- rbinom(N, 1, X1.p)
X2 <- runif(N, -1, 1)
cmp <- factor(ifelse(X1 == 0 & X2 >=-0.5 & X2 <= 0.5, "DE", 
                     ifelse(X2 >= 0.98, "AT",
                            ifelse(X2 <= -0.99, "NT", "CO"))))

ATE.true <- (1 - X1.p)*0.5*0.075 + 0.01 * 0.05 + 0.02 * 0.15 +
  (1 - 0.02 - (1 - X1.p)*0.5)*0.125

## covariate-adjusted BP bounds
pi00.0 <- ((cmp == "NT") * (1 - beta.NT.0(X1, X2)) + 
             (cmp == "CO") * (1 - beta.CO.0(X1, X2))) 
pi01.0 <- ((cmp == "AT") * (1 - beta.AT.d(X1, X2)) + 
             (cmp == "DE") * (1 - beta.DE.d(X1, X2))) 
pi10.0 <- ((cmp == "NT") * beta.NT.0(X1, X2) + 
             (cmp == "CO") * beta.CO.0(X1, X2)) 
pi11.0 <- ((cmp == "AT") * beta.AT.d(X1, X2) + 
             (cmp == "DE") * beta.DE.d(X1, X2)) 
pi00.1 <- ((cmp == "NT") * (1 - beta.NT.0(X1, X2)) + 
             (cmp == "DE") * (1 - beta.DE.0(X1, X2))) 
pi01.1 <- ((cmp == "AT") * (1 - beta.AT.d(X1, X2)) + 
             (cmp == "CO") * (1 - beta.CO.d(X1, X2))) 
pi10.1 <- ((cmp == "NT") * beta.NT.0(X1, X2) + 
             (cmp == "DE") * beta.DE.0(X1, X2)) 
pi11.1 <- ((cmp == "AT") * beta.AT.d(X1, X2) + 
             (cmp == "CO") * beta.CO.d(X1, X2))

pi <- cbind(pi00.0, pi01.0, pi10.0, pi11.0,
            pi00.1, pi01.1, pi10.1, pi11.1)

cov.lower <- apply(pi, 1, gamma.l)
cov.upper <- apply(pi, 1, gamma.u)

pi.mean <- c(mean(pi00.0), mean(pi01.0), mean(pi10.0), mean(pi11.0),
             mean(pi00.1), mean(pi01.1), mean(pi10.1), mean(pi11.1))
mean.lower <- gamma.l(pi.mean)
mean.upper <- gamma.u(pi.mean)

c(mean(cov.lower), mean(cov.upper)) ## theoretical covariate-assisted BP bounds
c(mean.lower, mean.upper) ## theoretical covariate-agnostic BP bounds

par(mar = c(3,4,1,1))
plot(NULL, ylim = c(ATE.true - 0.24,ATE.true + 0.21), xlim = c(0.2,0.8), 
     xlab = "", ylab = "Average Treatment Effect", xaxt = 'n')
axis(1, at = c(0.35, 0.65), labels = c("Covariate-Agnostic", "Covariate-Adjusted"))
abline(h = ATE.true, col = 'red', lty = 'dashed')
abline(h = 0)
arrows(x0=0.33, y0=mean.lower, x1=0.33, y1=mean.upper, 
       code=3, angle=90, length=0.05, lwd=2,col = 'blue')
arrows(x0=0.63, y0=mean(cov.lower)-0.0007, x1=0.63, y1=mean(cov.upper)+0.0007, 
       code=3, angle=90, length=0.05, lwd=2, col = 'blue')

## Simulated data
set.seed(956)
n <- 5000

X1 <- rbinom(n, 1, X1.p)
X2 <- runif(n, -1, 1)
cmp <- factor(ifelse(X1 == 0 & X2 >=-0.5 & X2 <= 0.5, "DE", "CO"))

Z <- rbinom(n, 1, 0.5)
A <- ifelse(cmp == "NT", 0,
            ifelse(cmp == "AT", 1,
                   ifelse(cmp == "CO", Z, 1 - Z)))
Y.0 <- ifelse(cmp == "NT", rbinom(n, 1, beta.NT.0(X1, X2)),
              ifelse(cmp == "AT", rbinom(n, 1, beta.AT.0(X1, X2)),
                     ifelse(cmp == "CO", rbinom(n, 1, beta.CO.0(X1, X2)),
                            rbinom(n, 1, beta.DE.0(X1, X2)))))
Y.1 <- ifelse(cmp == "NT", rbinom(n, 1, beta.NT.d(X1, X2)),
              ifelse(cmp == "AT", rbinom(n, 1, beta.AT.d(X1, X2)),
                     ifelse(cmp == "CO", rbinom(n, 1, beta.CO.d(X1, X2)),
                            rbinom(n, 1, beta.DE.d(X1, X2)))))
Y <- A*Y.1 + (1 - A)*Y.0

dat <- cbind.data.frame(X1, X2, Z, A, Y)
dat$Y.A <- factor(dat$Y):factor(dat$A)

## Estimated covariate-agnostic BP-bounds
## compute estimated probabilities
pi.hat <- c(mean((1 - dat$A[dat$Z == 0]) * (1 - dat$Y[dat$Z == 0])),
            mean(dat$A[dat$Z == 0] * (1 - dat$Y[dat$Z == 0])),
            mean((1 - dat$A[dat$Z == 0]) * dat$Y[dat$Z == 0]),
            mean(dat$A[dat$Z == 0] * dat$Y[dat$Z == 0]),
            mean((1 - dat$A[dat$Z == 1]) * (1 - dat$Y[dat$Z == 1])),
            mean(dat$A[dat$Z == 1] * (1 - dat$Y[dat$Z == 1])),
            mean((1 - dat$A[dat$Z == 1]) * dat$Y[dat$Z == 1]),
            mean(dat$A[dat$Z == 1] * dat$Y[dat$Z == 1]))

## natural bounds
round(c(p.l1(pi.hat), p.u1(pi.hat)), 3)

## Balke & Pearl bounds
round(c(gamma.l(pi.hat), gamma.u(pi.hat)), 3)

## bootstrap CI's
B <- 1000
boot.res <- matrix(NA, nrow = B, ncol = 2)
for (b in 1:B) {
  dat.b <- dat[sample(1:n, n, replace = T),]
  pi.hat.b <- c(mean((1 - dat.b$A[dat.b$Z == 0]) * (1 - dat.b$Y[dat.b$Z == 0])),
                mean(dat.b$A[dat.b$Z == 0] * (1 - dat.b$Y[dat.b$Z == 0])),
                mean((1 - dat.b$A[dat.b$Z == 0]) * dat.b$Y[dat.b$Z == 0]),
                mean(dat.b$A[dat.b$Z == 0] * dat.b$Y[dat.b$Z == 0]),
                mean((1 - dat.b$A[dat.b$Z == 1]) * (1 - dat.b$Y[dat.b$Z == 1])),
                mean(dat.b$A[dat.b$Z == 1] * (1 - dat.b$Y[dat.b$Z == 1])),
                mean((1 - dat.b$A[dat.b$Z == 1]) * dat.b$Y[dat.b$Z == 1]),
                mean(dat.b$A[dat.b$Z == 1] * dat.b$Y[dat.b$Z == 1]))
  boot.res[b,] <- c(gamma.l(pi.hat.b), gamma.u(pi.hat.b))
}

arrows(x0=0.37, y0=quantile(boot.res[,1],0.025), 
       x1=0.37, y1=quantile(boot.res[,2],0.975), 
       code=3, angle=90, length=0.05, lwd=2,col = 'darkgreen')

## Estimated covariate-assisted BP-bounds
M <- 10 ## number of data splits
test.indices <- list()
remaining <- 1:n
for (m in 1:(M-1)) {
  test.indices[[m]] <- sample(remaining, floor(n/M), replace = F)
  remaining <- remaining[! (remaining %in% test.indices[[m]])]
}
test.indices[[M]] <- remaining

analysis <- function(dat.train, dat.test) {
  
  pi.ya.0.rf <- rpart(Y.A ~ X1 + X2, data = dat.train[dat.train$Z==0, ])
  dat.test <- cbind.data.frame(dat.test, predict(pi.ya.0.rf, type = 'prob',
                                                 newdata = dat.test))
  ## compile pi.hat_{ya.0} estimates
  colnames(dat.test)[(ncol(dat.test) - 3):ncol(dat.test)] <-
    c("pi_00.0", "pi_01.0", "pi_10.0", "pi_11.0")
  
  pi.ya.1.rf <- rpart(Y.A ~ X1 + X2, data = dat.train[dat.train$Z==1, ])
  dat.test <- cbind.data.frame(dat.test, predict(pi.ya.1.rf, type = 'prob',
                                                 newdata = dat.test))
  ## compile pi.hat_{ya.1} estimates
  colnames(dat.test)[(ncol(dat.test) - 3):ncol(dat.test)] <-
    c("pi_00.1", "pi_01.1", "pi_10.1", "pi_11.1")
  
  lambda.1 <- 0.5
  ## compile all the IF contributions
  dat.test$psi_00.0 <- dat.test$pi_00.0 + 
    (1 - dat.test$Z) * ((1 - dat.test$Y) * 
                          (1 - dat.test$A) - dat.test$pi_00.0) / 
    (1 - lambda.1)
  dat.test$psi_01.0 <- dat.test$pi_01.0 + 
    (1 - dat.test$Z) * ((1 - dat.test$Y) * dat.test$A - dat.test$pi_01.0) / 
    (1 - lambda.1)
  dat.test$psi_10.0 <- dat.test$pi_10.0 + 
    (1 - dat.test$Z) * (dat.test$Y * (1 - dat.test$A) - dat.test$pi_10.0) / 
    (1 - lambda.1)
  dat.test$psi_11.0 <- dat.test$pi_11.0 + 
    (1 - dat.test$Z) * (dat.test$Y * dat.test$A - dat.test$pi_11.0) / 
    (1 - lambda.1)
  dat.test$psi_00.1 <- dat.test$pi_00.1 + 
    dat.test$Z * 
    ((1 - dat.test$Y) * (1 - dat.test$A) - dat.test$pi_00.1) / lambda.1
  dat.test$psi_01.1 <- dat.test$pi_01.1 + 
    dat.test$Z * ((1 - dat.test$Y) * dat.test$A - dat.test$pi_01.1) / lambda.1
  dat.test$psi_10.1 <- dat.test$pi_10.1 + 
    dat.test$Z * (dat.test$Y * (1 - dat.test$A) - dat.test$pi_10.1) / lambda.1
  dat.test$psi_11.1 <-dat.test$pi_11.1 + 
    dat.test$Z * (dat.test$Y * dat.test$A - dat.test$pi_11.1) / lambda.1
  
  pi.hats <- dat.test[,(ncol(dat.test) - 15):(ncol(dat.test) - 8)]
  psi.hats <- dat.test[,(ncol(dat.test) - 7):ncol(dat.test)]
  p.l.hats <- cbind(p.l1(pi.hats)[,1], p.l2(pi.hats)[,1], p.l3(pi.hats)[,1],
                    p.l4(pi.hats)[,1], p.l5(pi.hats)[,1], p.l6(pi.hats)[,1],
                    p.l7(pi.hats)[,1], p.l8(pi.hats)[,1])
  p.u.hats <- cbind(p.u1(pi.hats)[,1], p.u2(pi.hats)[,1], p.u3(pi.hats)[,1],
                    p.u4(pi.hats)[,1], p.u5(pi.hats)[,1], p.u6(pi.hats)[,1],
                    p.u7(pi.hats)[,1], p.u8(pi.hats)[,1])
  L.hats <- cbind(p.l1(psi.hats)[,1], p.l2(psi.hats)[,1], p.l3(psi.hats)[,1],
                  p.l4(psi.hats)[,1], p.l5(psi.hats)[,1], p.l6(psi.hats)[,1],
                  p.l7(psi.hats)[,1], p.l8(psi.hats)[,1])
  U.hats <- cbind(p.u1(psi.hats)[,1], p.u2(psi.hats)[,1], p.u3(psi.hats)[,1],
                  p.u4(psi.hats)[,1], p.u5(psi.hats)[,1], p.u6(psi.hats)[,1],
                  p.u7(psi.hats)[,1], p.u8(psi.hats)[,1])
  
  argmax.p.l <- apply(p.l.hats, MARGIN = 1, FUN = which.max)
  
  argmin.p.u <- apply(p.u.hats, MARGIN = 1, FUN = which.min)
  
  IF.l <- sapply(1:nrow(dat.test), function(i) { L.hats[i, argmax.p.l[i]] }, 
                 simplify = 0)
  IF.u <- sapply(1:nrow(dat.test), function(i) { U.hats[i, argmin.p.u[i]] }, 
                 simplify = 0)
  
  ## IF-based Balke & Pearl bounds
  
  return(c(lower = mean(IF.l), upper = mean(IF.u),
           lower.var = var(IF.l), upper.var = var(IF.u)))
  
}

results.BP <- lapply(test.indices, function(inds) {
  analysis(dat.train = dat[-inds,], dat.test = dat[inds,])
})
BP.low <- mean(sapply(results.BP, function(x) {x[1]}, simplify = 0))
BP.upp <- mean(sapply(results.BP, function(x) {x[2]}, simplify = 0))
BP.low.var <- mean(sapply(results.BP, function(x) {x[3]}, simplify = 0))
BP.upp.var <- mean(sapply(results.BP, function(x) {x[4]}, simplify = 0))

arrows(x0=0.67, y0=BP.low + qnorm(0.025) * sqrt(BP.low.var / n), 
       x1=0.67, y1=BP.upp + qnorm(0.975) * sqrt(BP.upp.var / n), 
       code=3, angle=90, length=0.05, lwd=2, col = 'darkgreen')


\end{verbatim}


\end{appendices}

\end{document}